\documentclass[3p]{elsarticle}

\usepackage[utf8]{inputenc}
\usepackage{mathtools}
\usepackage{amsthm}
\usepackage{amssymb}
\usepackage{amsmath}
\usepackage{amsfonts}
\usepackage{epsfig}
\usepackage[mathscr]{eucal}
\usepackage[ruled,vlined]{algorithm2e}
\usepackage{algorithmic}
\usepackage{enumitem}

\usepackage{setspace,color}
\usepackage{tikz}
\usepackage{hyperref}
\usepackage[capitalize]{cleveref}
\usepackage{numcompress}
\usepackage{multirow}
\usepackage{subfigure}

\usepackage[numbers]{natbib}
\usepackage{stackrel}

\usetikzlibrary{snakes}
\usetikzlibrary{arrows,shapes}

\DeclarePairedDelimiterX{\inp}[2]{\langle}{\rangle}{#1, #2}
\newcommand\op[1]{\left\|#1\right\|}
\newcommand\fro[1]{\left\| #1 \right\|_{\rm F}}
\newcommand\nuc[1]{\left\| #1 \right\|_{\ast}}
\newcommand\bnorm[1]{\left\|#1\right\|_{\bmB,2}}
\newcommand\binfnorm[1]{\left\|#1\right\|_{\bmB,\infty}}

\def\vec{\textsf{vec}}

\newcommand{\bomega}{\boldsymbol{\omega}}

\newcommand{\bSigma}{\boldsymbol{\Sigma}}
\newcommand{\bGamma}{\boldsymbol{\Gamma}}

\newcommand{\ltwo}[1]{\|#1\|_2}

     \def\EE{\mathbb{E}}

     \def\MM{\mathbb{M}}

     \def\RR{\mathbb{R}}
     
     \def\TT{\mathbb{T}}
     
     \def\VV{\mathbb{V}}

\newcommand{\bLa}{\boldsymbol{\Lambda}}

\newcommand{\bOm}{\boldsymbol{\Omega}}

\newcommand{\rank}{\text{rank}}
\newcommand{\range}{\text{range}}

\def\bdiag{\textsf{bdiag}}

\def\wt{\widetilde}
\def\no{\notag}
\def\bzero{\boldsymbol{0}}
\def\dim{\textsf{dim}}

\def\supp{\textsf{supp}}

\def\calP{{\cal  P}}

\def\mB{{\mathcal B}}

\def\mG{{\mathcal G}}
\def\mH{{\mathcal H}}
\def\mI{{\mathcal I}}

\def\mP{{\mathcal P}}

\def\mZ{{\mathcal Z}}

\def\bmB{{\boldsymbol \mB}}

\def\bmG{{\boldsymbol \mG}}
\def\bmH{{\boldsymbol \mH}}
\def\bmI{{\boldsymbol \mI}}

\def\bmP{{\boldsymbol \mP}}

\def\bmZ{{\boldsymbol \mZ}}

\def\A{{\boldsymbol A}}
\def\B{{\boldsymbol B}}
\def\C{{\boldsymbol C}}
\def\D{{\boldsymbol D}}

\def\F{{\boldsymbol F}}

\def\H{{\boldsymbol H}}
\def\I{{\boldsymbol I}}

\def\L{{\boldsymbol L}}
\def\M{{\boldsymbol M}}

\def\R{{\boldsymbol R}}
\def\S{{\boldsymbol S}}

\def\U{{\boldsymbol U}}
\def\V{{\boldsymbol V}}
\def\W{{\boldsymbol W}}
\def\X{{\boldsymbol X}}
\def\Y{{\boldsymbol Y}}
\def\Z{{\boldsymbol Z}}

\def\b{{\boldsymbol b}}

\def\e{{\boldsymbol e}}

\def\k{{\boldsymbol k}}
\def\l{{\boldsymbol l}}

\def\x{{\boldsymbol x}}
\def\y{{\boldsymbol y}}
\def\z{{\boldsymbol z}}

\def\aal{\boldsymbol{\alpha}}
\def\bbe{\boldsymbol{\beta}}

\def\gga{\boldsymbol{\gamma}}

\def\oom{\boldsymbol{\omega}}

\def\La{{\boldsymbol\Lambda}}

\def\Sig{\boldsymbol{\Sigma}}

\def\wt{\widetilde}

\newtheorem{thm}{Theorem}[section]

\newtheorem{lem}[thm]{Lemma}

\newdefinition{defi}[thm]{Definition}
\newdefinition{rmk}[thm]{Remark}
\newdefinition{notation}[thm]{Notation}
\newdefinition{example}[thm]{Example}
\newtheorem{assump}[thm]{Assumption}
\newdefinition{obs}[thm]{Observation}

\newproof{pf}{Proof}
\newproof{poth4}{Proof of \cref{Th4}}

\numberwithin{equation}{section}

\crefformat{equation}{\textup{#2(#1)#3}}
\crefrangeformat{equation}{\textup{#3(#1)#4--#5(#2)#6}}
\crefmultiformat{equation}{\textup{#2(#1)#3}}{ and \textup{#2(#1)#3}}
{, \textup{#2(#1)#3}}{, and \textup{#2(#1)#3}}
\crefrangemultiformat{equation}{\textup{#3(#1)#4--#5(#2)#6}}%
{ and \textup{#3(#1)#4--#5(#2)#6}}{, \textup{#3(#1)#4--#5(#2)#6}}{, and \textup{#3(#1)#4--#5(#2)#6}}

\Crefformat{equation}{#2Equation~\textup{(#1)}#3}
\Crefrangeformat{equation}{Equations~\textup{#3(#1)#4--#5(#2)#6}}
\Crefmultiformat{equation}{Equations~\textup{#2(#1)#3}}{ and \textup{#2(#1)#3}}
{, \textup{#2(#1)#3}}{, and \textup{#2(#1)#3}}
\Crefrangemultiformat{equation}{Equations~\textup{#3(#1)#4--#5(#2)#6}}%
{ and \textup{#3(#1)#4--#5(#2)#6}}{, \textup{#3(#1)#4--#5(#2)#6}}{, and \textup{#3(#1)#4--#5(#2)#6}}

\crefdefaultlabelformat{#2\textup{#1}#3}

\title{Restoration Guarantee of Image Inpainting via Low Rank Patch Matrix Completion}

\author[HKUST]{Jian-Feng Cai\fnref{JFC}}
\ead{jfcai@ust.hk}
\author[TJU]{Jae Kyu Choi\fnref{JKC}}
\ead{jaycjk@phantomics.io}
\author[HKUST]{Jingyang Li\fnref{JFC}\corref{cor}}
\ead{jlieb@connect.ust.hk}
\author[SZU]{Guojian Yin\fnref{GY}}
\ead{yin@szu.edu.cn}

\cortext[cor]{Corresponding author}

\address[HKUST]{Department of Mathematics, The Hong Kong University of Science and Technology, Hong Kong}
\address[TJU]{Phantomics Inc, Seoul, 07788, Republic of Korea.}
\address[SZU]{Institute for Advanced Study, Shenzhen University, Shenzhen, Guangdong, China}

\fntext[JFC]{J. F. Cai and J. Li are supported by the Hong Kong Research Grants Council (HKRGC) GRF 16306821, 16309219, and 16310620.}
\fntext[JKC]{J. K. Choi is supported in part by Korea Institute for Advancement of Technology P0021355.}

\begin{document}
	
	\begin{abstract}
		In recent years, patch-based image restoration approaches have demonstrated superior performance compared to conventional variational methods. This paper delves into the mathematical foundations underlying patch-based image restoration methods, with a specific focus on establishing restoration guarantees for patch-based image inpainting, leveraging the assumption of self-similarity among patches.
		To accomplish this, we present a reformulation of the image inpainting problem as structured low-rank matrix completion, accomplished by grouping image patches with potential overlaps. By making certain incoherence assumptions, we establish a restoration guarantee, given that the number of samples exceeds the order of $r\log^2(N)$, where $N\times N$ denotes the size of the image and $r>0$ represents the sum of ranks for each group of image patches.
		Through our rigorous mathematical analysis, we provide valuable insights into the theoretical foundations of patch-based image restoration methods, shedding light on their efficacy and offering guidelines for practical implementation.
	\end{abstract}
	
	
	\maketitle
	
	\pagestyle{myheadings}
	\thispagestyle{plain}
	\markboth{Jian-Feng Cai, Jae Kyu Choi, and Jingyang Li}{Restoration Guarantee of Image Inpainting via Low Rank Patch Matrix Completion}
\section{Introduction}
Image restoration is a widely studied inverse problem that aims to restore a high-quality image from a degraded version caused by various factors such as imaging, acquisition, and communication processes. Due to the ill-posed nature of linear inverse problems encountered in image restoration, it is essential to incorporate prior information about the target image to achieve a high-quality recovery. Consequently, before delving into image restoration tasks, it is natural and fundamental to address the question of identifying an appropriate model or descriptor for images.
Several well-established models have been widely used in image restoration, including variational approaches such as the Rudin-Osher-Fatemi model \cite{L.I.Rudin1992}, the inf-convolution model \cite{A.Chambolle1997}, and the total generalized variation (TGV) model \cite{K.Bredies2014,K.Bredies2010}. In addition to these variational approaches, applied harmonic analysis methods have also been introduced, such as curvelets \cite{E.Candes2006}, Gabor frames \cite{H.Ji2017,Mallat2008}, shearlets \cite{G.Kutyniok2011}, complex tight framelets \cite{B.Han2014}, and wavelet frames \cite{C.Bao2016,J.F.Cai2010,J.F.Cai2009/10}. The common underlying concept in these approaches is to seek a sparse approximation of images by employing a pre-designed linear transformation in conjunction with a sparsity-promoting regularization term. It is worth noting that the thresholding algorithms inherent in these sparse regularization techniques can be linked to nonlinear evolution partial differential equations (PDEs) \cite{B.Dong2017}. As a result, these approaches aim to preserve the sharpness of image singularities while simultaneously regularizing smooth regions.

While the connections between variational approaches and applied harmonic analysis approaches have expanded the scope of mathematical modeling and numerical algorithms \cite{J.F.Cai2012,J.K.Choi2020,B.Dong2017a}, it is important to recognize that these approaches typically assume a target image to be a piecewise smooth function. Consequently, they excel in restoring images with prominent cartoon-like structures. However, natural images are incredibly diverse, with each image containing its unique textures, features, and, in addition, the cartoon-like components. These additional features cannot be adequately modeled or approximated as piecewise smooth functions.
In the literature, there are some variational approaches \cite{S.Osher2003,L.A.Vese2003} and applied harmonic analysis approaches \cite{J.F.Cai2009/10,H.Schaeffer2013} that aim to incorporate texture modeling simultaneously. However, when we endeavor to model or regularize each layer of an image based on its distinct properties, these model-based approaches may fall short in capturing the true diversity of natural images.

Recently, a plethora of image restoration approaches have emerged, employing adaptive techniques on image patches either within the target image itself or across a database of images.  
%Denoising methods like non-local means \cite{A.Buades2005}, BM3D \cite{K.Dabov2006}, dictionary learning \cite{M.Elad2006}, weighted nuclear norm minimization \cite{S.Gu2014} exploited the self-similarity within patches from the image.
Existing denoising methods \cite{A.Buades2005,K.Dabov2006,M.Elad2006,S.Gu2014} exploited the self-similarity within patches from the image.
In particular, similar patches are assumed to lie on some low dimensional subspace \cite{H.Ji2010} or low dimensional manifold \cite{S.Osher2017, laus2017nonlocal}.
%Patches also exhibit redundancy at various scales, enabling classical patch-based methods to create diverse images from a single source image through slight randomness at low resolution, outperforming single-image GANs \cite{glasner2009super, granot2022drop}. 
Another approach approximates patch distributions using parametric models, like the non-local Bayesian method assuming Gaussian patch distributions \cite{lebrun2013nonlocal} or Gaussian mixture model \cite{zoran2011learning}.
%Moreover, patches in natural images exhibit redundancy at various scales, leading to classical multi-image super-resolution within the same scale and example-based super-resolution across different scales \cite{glasner2009super}.
%Notably, classical patch-based methods can unconditionally produce diverse images from a single natural image by introducing slight randomness at low resolution, surpassing single-image GAN \cite{granot2022drop}.
%Another line of research approximates the patch distribution within certain parametric families. 
%A non-local Bayesian patch-based method was proposed in \cite{lebrun2013nonlocal}, where the patches are assumed to follow a Gaussian distribution. 
%Alternatively, one may minimizing the sum of data fidelity term and negative log likelihood of image patches under certain prior distribution such as Gaussian mixture model \cite{zoran2011learning}.
%Notice when the image is contaminated by Gaussian noise, these two approaches are equivalent with the same choice of prior. 
%More recently, 
Interestingly, convolution neural networks (CNNs) can also be viewed as a patch-based method that mysteriously exploits the underlying patch structures. Despite of its complication in theoretical explanation, it has demonstrated remarkable success and widespread utility across various image processing tasks such as image classification \cite{Hou_2016_CVPR,sharma2017patch}, segmentation \cite{ronneberger2015u}, and denoising \cite{zhang2019patch}. However, a notable limitation of these deep learning-based approaches lies in their substantial data requirements.
To this end, \cite{altekruger2022patchnr} proposed a patch normalizing flow regularizer trained on patches from a limited number of clean images.
The success of CNNs, is attributed to the implicit regularization capabilities embedded within CNNs when dealing with image patches.
This is also the idea of deep image prior framework, as proposed by \cite{ulyanov2018deep}, which asserts that the architecture of a generator network alone is sufficiently adept at capturing a substantial portion of low-level image statistics, even without training data.

These patch-based approaches have demonstrated significant improvements in restoration quality compared to conventional variational methods. The foundation of these approaches lies in the assumption of self-similar patches. Exploiting this self-similarity allows for the inference of information about degraded regions from remote regions, enabling the utilization of image structures more effectively during the restoration process. Despite the recent advancements in modeling and numerical algorithms, it is worth noting that there is a paucity of mathematical analysis concerning patch-based image restoration approaches in the current body of literature.

To fill this gap, we consider a simple model with the focus on image inpainting, which can be formulated as the following optimization problem:
\begin{align}\label{eq:iminp}
	\text{find}~~\z~~\text{subject to}~~\y=\bmP_{\La}\z+\b,
\end{align}
where $\bOm$ represents the set of indices of the image, $\La\subseteq\bOm$ represents the set of indices where the observed image $\y$ is available, and $\b$ characterizes the noise. We assume that the indices in $\La$ are i.i.d. samples drawn from the uniform distribution over $\bOm$, denoting the cardinality of $\La$ as $m$. It is worth mentioning that we could also consider the scenario where $\La$ is drawn from all $m$-subsets of $\bOm$ to achieve the same error bound without the possibility of collisions \cite{Gross2011}. However, throughout this paper, we capitalize on the advantages of using i.i.d. samples drawn from $\bOm$. The task of inpainting $\z$ on $\bOm\setminus\La$ from the observed image $\y$ arises when certain pixels are randomly missing due to factors such as unreliable communication channels \cite{T.F.Chan2006}, dust spots or cracks in film \cite{H.Ji2011}, or corruption by salt-and-pepper noise \cite{R.H.Chan2005}.

In many patch-based image restoration methods, image patches are grouped based on their similarity or linear correlation. Specifically, we define a linear operator $\bmH$, which maps an $N\times N$ discrete image $\z$ into an $n^2\times N_p$ patch matrix $\bmH(\z)$ by concatenating its $n\times n$ patches into column vectors. Here, $N_p$ denotes the number of patches, which depends on the boundary conditions defining the image patches. Using two-dimensional multi-indices, we express this as:
\begin{align}\label{def:patchmtx}
	\bmH(\z)(\k,\l) = \z(\k+\l),
\end{align}
where all indices start from $0$ throughout this paper.
Assuming an underlying true image $\z^*$, we consider grouping $n\times n$ patches of $\z^*$ into $K$ groups, $\bOm_1,\cdots,\bOm_K\subseteq\bOm$, with possible overlaps based on their linear correlations. For each group $k=1,\cdots,K$, the matrix $\bmH(\z^*)\big|_{\bOm_k}$ can be well approximated by a low-rank matrix. Motivated by this, we propose the following nuclear norm minimization problem for patch-based image inpainting using the notation in \eqref{eq:iminp}:
\begin{align}\label{prob:nnmin}
	\min_{\z}~\sum_{k=1}^K\nuc{\bmH(\z)\big|_{\bOm_k}}~~~~\text{subject to}~~~\y=\bmP_{\La}\z+\b.
\end{align}
In this paper, our goal is to demonstrate that the problem in \cref{prob:nnmin} enables perfect restoration (when $\b=\mathbf{0}$) or stable restoration (when $\b$ is bounded noise) as long as the number of revealed entries exceeds the order of $r\log^2\left|\bOm\right|$, where $r=r_1+\cdots+r_K$, and $r_k=\rank\left(\bmH(\z^*)\big|_{\bOm_k}\right)$ for $k=1,\cdots,K$.

\subsection{Key observation and overview of related works}

Our mathematical analysis is based on the observation that most existing patch based approaches can be viewed as various matrix factorizations of structured matrices with low dimensional properties. More precisely, developing theory of these patch-based approaches are equivalent to developing theory of low dimensional structured matrix restoration problems. The low dimensionality is often manifested as a low-rank property, which allows us to leverage extensive studies on low-rank matrix restoration in the literature  \cite{J.F.Cai2010a,E.J.Candes2009,Y.Chi2019,Gross2011,Recht2011,B.Recht2010}. In fact, our model has close relation with the following existing models. 

\begin{enumerate}[label = $\bullet$]
	\item \textit{K-SVD}. 
	Consider the K-SVD method for image denoising \cite{M.Aharon2006}, which explicitly learns a dictionary from image patches, assuming sparse coefficients. Given a noisy image, the goal of K-SVD denoising is to restore a clear image $\z$ that satisfies the factorization $\bmH\left(\z\right)\approx\D\C^T$, where $\D\in\mathbb{R}^{n^2\times d}$ is a dictionary with unit length, and $\C\in\mathbb{R}^{N_p\times d}$ is a sparse coefficient matrix.
	
	\item \textit{Data-driven tight frame. }
	Another example is the data-driven tight frame \cite{J.F.Cai2014}, which aims to construct adaptive discrete wavelet tight frame filter banks from a given clean or noisy image $\z^*$. By relaxing the unitary extension principle and considering square orthogonal matrix constraints through filter concatenation, the data-driven tight frame aims for the factorization $\bmH\left(\z^*\right)\approx\D\C^T$, where $\D\in\RR^{n^2\times n^2}$ is an orthogonal matrix representing the tight frame filter banks, and $\C\in\RR^{N_p\times n^2}$ is a sparse matrix with each row corresponding to the sparse frame coefficients.
	
	\item \textit{Extension of superposition of complex sinusoids. }
	In various signal processing applications such as fluorescence microscopy, radar imaging, nuclear magnetic resonance spectrography, and others, the underlying signal can be expressed as a superposition of $r$ complex sinusoids \cite{J.F.Cai2018a,J.F.Cai2019}. These signals satisfy a linear recurrence formula (LRF), which implies that each patch can be expressed as a causal linear combination of its precedent patches. Consequently, the resulting Hankel matrix satisfies $\rank\left(\bmH\left(\z^*\right)\right)\leq r$. Therefore, the singular value decomposition of $\bmH\left(\z^*\right)$ yields a data-driven tight frame factorization $\bmH\left(\z^*\right)=\D\C^T$, where $\D$ is an orthogonal matrix, and $\C$ is a frame coefficient matrix with at most $r$ nonzero row vectors \cite{J.F.Cai2022,J.F.Cai2020}.
\end{enumerate}

Indeed, many patch-based approaches employ grouping strategies to exploit the linear correlation among image patches. The works in \cite{S.Gu2014,H.Ji2010,L.Ma2017}, for example, collect similar patches into a matrix and promote (weighted) nuclear norm minimization to restore a patch matrix, ensuring that each grouped submatrix is of low rank. Let $\bOm_1,\cdots,\bOm_K\subseteq\bOm$ be $K$ possibly overlapped groups, where $\bmH\left(\z^*\right)\big|_{\bOm_k}$ can be well approximated by a matrix of rank $r_k\ll\min\left\{n^2,|\bOm_k|\right\}$ for $k=1,\cdots,K$. With this approximation, we have $\bmH\left(\z^*\right)\big|_{\bOm_k}\approx\D_k\C_k^T$ where $\D_k\in\RR^{n^2\times r_k}$ and $\C_k\in\RR^{|\bOm_k|\times r_k}$. This leads to the following factorization:
\begin{align}\label{eq:groupfac}
	\H^*=\left[\begin{array}{ccc}
		\bmH\left(\z^*\right)\big|_{\bOm_1}&&\\
		&\ddots&\\
		&&\bmH\left(\z^*\right)\big|_{\bOm_K}
	\end{array}\right]\approx\left[\begin{array}{ccc}
		\D_1&&\\
		&\ddots&\\
		&&\D_K
	\end{array}\right]\left[\begin{array}{ccc}
		\C_1^T&&\\
		&\ddots&\\
		&&\C_K^T
	\end{array}\right]:=\D\C^T.
\end{align}
This representation can be seen as a dictionary representation, with the sparsity pattern of $\C$ determined by the grouping.

Additionally, some patch-based approaches assume that image patches can be embedded into a low-dimensional manifold \cite{S.Osher2017}. In this case, neighboring patches, according to the induced metric on the manifold, can be well approximated by a tangent space of the manifold, which corresponds to a low-rank matrix. Therefore, similar to the previous case, the grouping argument can be applied to obtain the factorization \cref{eq:groupfac}.

\subsection{Organization and notation of paper}

The rest of this paper is organized as follows. We begin with our patch based image inpainting model in consideration and our main results in Section \ref{sec:modelmainresults}. Proofs of our main results are given in Section \ref{sec:proofs}. The performance of the inpainting model is demonstrated by numerical experiments in Section \ref{sec:numerical}. Finally, Section \ref{sec:conclusion} concludes this paper with some relevant future directions.

Throughout this paper, all two dimensional discrete images will be denoted by the bold faced lower case letters. Note that we consider a two dimensional image as a real valued function defined on the square grid $\bOm=\left\{0,\cdots,N-1\right\}^2$, which can also be identified with an $N^2\times1$ vector whenever convenient. In this sequel, the $\ell_p$ norm of an image $\z$ is defined as
\begin{align}\label{finitelp}
	\left\|\z\right\|_p^p&=\sum_{\k\in\bOm}\left|\z(\k)\right|^p
\end{align}
for $p\in[1,\infty)$, and
\begin{align}\label{linfty}
	\left\|\z\right\|_{\infty}&=\max_{\k\in\bOm}\left|\z(\k)\right|.
\end{align}
For $p=0$, the $\ell_0$ seminorm encodes the number of nonzero entries:
\begin{align}\label{l0}
	\left\|\z\right\|_0=\#\left\{\k\in\bOm:\z(\k)\neq0\right\}.
\end{align}
Meanwhile, all matrices will be denoted by the bold faced upper case letters, and the $i$th row and the $j$th column of a matrix $\Z$ will be denoted by $\Z^{(i,:)}$ and $\Z^{(:,j)}$, respectively, and the $(i,j)$th entry of a matrix $\Z$ will be denoted by $\Z_{ij}$. In addition, $\op{\Z}$, $\fro{\Z}$, $\left\|\Z\right\|_{\infty}$, and $\left\|\Z\right\|_0$ denote the spectral norm, Frobenius norm, the maximum magnitude of the entries, and the number of nonzero entries, respectively. For matrices $\Z_1$ and $\Z_2$, the inner product is defined as $\left\langle\Z_1,\Z_2\right\rangle=\textsf{trace}\left(\Z_2^T\Z_1\right)$, and for images $\z_1$ and $\z_2$, the inner product is defined as
\begin{align}\label{l2inner}
	\left\langle\z_1,\z_2\right\rangle=\sum_{\k\in\bOm}\z_1(\k)\z_2(\k).
\end{align}
For positive integer $N$, we denote $\left[N\right]=\left\{0,\cdots,N-1\right\}$; we will write $\bOm=\left\{0,\cdots,N-1\right\}^2=\left[N\right]^2$ in what follows.

All operators will be denoted by bold faced calligraphic letters. For instance, $\bmI$ denotes the identity operator, while $\I$ denotes the identity matrix. In addition, recall that $\bmH$ is an operator which maps an $N\times N$ discrete image $\z$ to an $n^2\times N_p$ matrix $\bmH\left(\z\right)$ by concatenating all $n\times n$ patches of $\z$ into column vectors. Notice that the definition \cref{def:patchmtx} of $\bmH\left(\z\right)$ involves the boundary condition defining $n\times n$ patches of $\z$. For instance, when $N_p=(N-n+1)^2$ (i.e. we ignore patches involving the boundary extensions), $\bmH\left(\z\right)$ is a two-fold Hankel matrix. Nevertheless, since the definition shares the property that $\bmH\left(\z\right)$ is constant along the antidiagonals (in the sense of two dimensional multi-indices) independent of the boundary condition, we will continue to call $\bmH\left(\z\right)$ a Hankel matrix without ambiguity. Finally, we use $C$, $C_1$, $c$, $c_0$, $c_1$, etc. to denote absolute constants whose values may be different from line to line.

\section{Model and main results}\label{sec:modelmainresults}
In this paper, we consider the true image $\z^* \in \mathbb{R}^{\bOm}$, where $\bOm = [N]^2$ represents the set of indices for the image pixels. However, only partial information of $\z^*$ is available for restoration. Let $\La = \{\k_1, \ldots, \k_m\} \subseteq \bOm$ be the set of indices where the information of $\z^*$ is available. Typically, the image inpainting problem is formulated as follows:
\begin{align}\label{prob:imginp}
	\text{Find}~\z~\text{such that}~~\bmP_{\La}(\z) = \y,
\end{align}
where $\y$ represents the partial measurements of $\z^*$ given by:
\begin{align}\label{noisefreemeas}
	\y = \bmP_{\La}(\z^*),
\end{align}
and $\bmP_{\La}(\z) = \sum_{i = 1}^m \z(\k_i)\e_{k_{i,1}}\e_{k_{i,2}}^T$. The goal is to recover the original image $\z^*$ from the available measurements $\y$.

Generally speaking, it is impossible to restore $\z^*$ from its partial entries $\bmP_{\La}\left(\z^*\right)$ as we can take arbitrary values on the missing region $\bOm\setminus\La$ without violating the equality constraint in \cref{prob:imginp}. Throughout this paper, we will consider the following prior of $\z^*$ in the patch domain. Recall that $\bmH\left(\z^*\right)$ is an $n^2\times N_p$ two-fold Hankel matrix generated by concatenating all $n\times n$ patches of $\z^*$ into column vectors. Let $\bOm_1,\cdots,\bOm_K\subseteq\bOm$ be $K$ possibly overlapped groups such that patches in each $\bOm_k$ are linearly correlated. Assume for simplicity that $\bOm_1,\cdots,\bOm_K$ are given and fixed with $|\bOm_k|=N_p$. Then each $n^2\times N_p$ matrix $\bmH\left(\z^*\right)\big|_{\bOm_k}$ is of low rank, say $\rank\left(\bmH(\z^*)\right)\big|_{\bOm_k}=r_k<\min\left\{n^2,N_p\right\}$. Hence, instead of directly restoring $\z^*$ directly, we can restore the $n^2\times N_p$ two-fold Hankel matrix such that each submatrix $\bmH\left(\z^*\right)\big|_{\bOm_k}$ is of low rank. This leads us to solve
\begin{align}\label{rankmin:noisefree}
	\min_{\z}\sum_{k=1}^K\rank\left(\bmH\left(\z\right)\big|_{\bOm_k}\right)~~\text{subject to}~~\bmP_{\La}\left(\z\right)=\y.
\end{align}
In practice, the partial measurement $\y$ can contain noise, i.e.
\begin{align}\label{noisymeas}
	\y=\bmP_{\La}\left(\z^*\right)+\b.
\end{align}
In this case, we solve
\begin{align}\label{rankmin:noisy}
	\min_{\z}\sum_{k=1}^K\rank\left(\bmH\left(\z\right)\big|_{\bOm_k}\right)~~\text{subject to}~~\left\|\bmP_{\La}\left(\z\right)-\y\right\|_2\leq\sqrt{m}\delta,
\end{align}
where $\delta\geq m^{-1/2}\ltwo{\b}$ characterizes the noise level.

The rank minimization problem is known to be NP-hard in general \cite{E.J.Candes2009}. However, the nuclear norm minimization has been widely recognized as an effective approach for approximating low-rank matrices \cite{B.Recht2010}. Therefore, we can consider solving the following alternative formulations to address the rank minimization problem.
For the noise-free case, where the partial measurement $\y$ is equal to $\bmP_{\La}(\z^*)$, we can solve the problem:
\begin{align}\label{prob:noisefree}
	\min_{\z} \sum_{k=1}^K \nuc{\bmH(\z)\big|_{\bOm_k}} \quad \text{subject to} \quad \bmP_{\La}(\z) = \y.
\end{align}
This formulation promotes low-rank submatrices within each group while satisfying the observed measurements constraint.

When noise is present in the partial measurement $\y$, i.e., $\y = \bmP_{\La}(\z^*) + \b$, we can solve the following problem to account for the noise:
\begin{align}\label{prob:noise}
	\min_{\z} \sum_{k=1}^K \nuc{\bmH(\z)\big|_{\bOm_k}} \quad \text{subject to} \quad \ltwo{\bmP_{\La}(\z) - \y} \leq \sqrt{m} \delta,
\end{align}
where $\delta = \|\b\|_{\infty}$. This formulation promotes low-rank submatrices while allowing for a bounded noise discrepancy between the observed measurements and the reconstructed image.
By solving either \cref{prob:noisefree} or \cref{prob:noise}, we aim to restore the two-fold Hankel matrix $\bmH(\z^*)$ with low-rank submatrices within each group while respecting the observed measurements constraint.

The effectiveness of low rank priors in image inpainting has been demonstrated in several works, such as \cite{S.Gu2014,H.Ji2010,L.Ma2017}, where significant improvements have been reported. These methods show that utilizing a low rank prior can effectively capture the underlying structure of an image and provide more reliable image representation. Furthermore, the theoretical aspects of low rank matrix completion have been extensively studied in the literature \cite{J.F.Cai2010a,J.F.Cai2018a,J.F.Cai2019,E.J.Candes2009,Y.Chi2019,Gross2011,Recht2011,B.Recht2010}.
Motivated by these theoretical studies, we provide a theoretical result that guarantees exact recovery of patch-based image inpainting when the set of observed indices $\bLa=\{\k_1,\cdots,\k_m\}\subseteq\bOm$ is randomly drawn from a uniform distribution over $\bOm$ in an i.i.d. manner.

Before presenting our main results, we introduce some assumptions regarding the true image $\z^*$ in this paper. To do so, we define $\H^*$ as follows:
\begin{align}\label{DefineH*}
	\H^*=\bdiag\left(\left\{\bmH\left(\z^*\right)\big|_{\bOm_k}:k=1,\cdots,K\right\}\right):=\left[\begin{array}{ccc}
		\bmH\left(\z^*\right)\big|_{\bOm_1}&&\\
		&\ddots&\\
		&&\bmH\left(\z^*\right)\big|_{\bOm_K}
	\end{array}\right].
\end{align}
For $k=1,\cdots,K$, let $\bmH\left(\z^*\right)\big|_{\bOm_k}=\U_k\Sig_k\V_k^T$ be the (compact) SVD, and we define
\begin{align}
	\U&=\bdiag\left(\left\{\U_k:k=1,\cdots,K\right\}\right),\label{USigmaV:U}\\
	\Sig&=\bdiag\left(\left\{\Sig_k:k=1,\cdots,K\right\}\right),\label{USigmaV:Sigma}\\
	\V&=\bdiag\left(\left\{\V_k:k=1,\cdots,K\right\}\right)\label{USigmaV:V}.
\end{align}
And thus $\H^* = \U\Sig\V^T$ is the SVD of $\H^*$. Using these notations, we summarize the assumptions as follows.
%	By \cref{DefineH*}, we have
%	\begin{align}\label{H*SVD}
	%		\H^*=\U\Sig\V^T=\left[\begin{array}{ccc}
		%			\U_1&&\\
		%			&\ddots&\\
		%			&&\U_K
		%		\end{array}\right]\left[\begin{array}{ccc}
		%			\Sig_1&&\\
		%			&\ddots&\\
		%			&&\Sig_K
		%		\end{array}\right]\left[\begin{array}{ccc}
		%			\V_1^T&&\\
		%			&\ddots&\\
		%			&&\V_K^T
		%		\end{array}\right],
	%	\end{align}
%	which is the SVD of $\H^*$. Under this notation, our assumption on $\z^*$ is summarized as follows.

\begin{assump}\label{assumptions} Assume that the true image $\z^*$ satisfies the followings.
	\begin{enumerate}[label={(\alph*)}, ref={\ref{assumptions}.(\alph*)}]
		\item \label{lowrankprior}
		For $k=1,\cdots,K$, $\bmH\left(\z^*\right)\big|_{\bOm_k}\in\RR^{n^2\times N_p}$ satisfies
		\begin{align*}
			\rank\left(\bmH\left(\z^*\right)\big|_{\bOm_k}\right)=r_k<\min\left\{n^2,N_p\right\},
		\end{align*}
		and we denote $r:=r_1+\cdots+r_K$.
		\item \label{assumptiononH}
		The operator $\bmH$ generating $n^2\times N_p$ patch matrix from $N\times N$ image satisfies
		\begin{align*}
			\max_{\k\in\bOm}\max_{i\in\left[n^2\right]}\left\|\bmH\left(\e_{k_1}\e_{k_2}^T\right)^{(i,:)}\right\|_0\leq4,~~\text{and}~~\max_{\k\in\bOm}\max_{j\in\left[N_p\right]}\left\|\bmH\left(\e_{k_1}\e_{k_2}\right)^{(:,j)}\right\|_0\leq4,
		\end{align*}
		where $\e_j$ denotes the $j$th standard basis of the Euclidean space $\mathbb{R}^N$.
		\item \label{pixeloccurrence}
		The $K$ groups $\bOm_1,\cdots,\bOm_K\subseteq\bOm$ satisfies
		\begin{align*}
			\max_{\k,\l\in\bOm}\frac{\sum_{k=1}^K\left\|\bmH\left(\e_{k_1}\e_{k_2}^T\right)\big|_{\bOm_k}\right\|_0}{\sum_{k=1}^K\left\|\bmH\left(\e_{l_1}\e_{l_2}^T\right)\big|_{\bOm_k}\right\|_0}\leq M
		\end{align*}
		for some absolute constant $0<M<+\infty$.
		\item \label{incoherence}
		(Incoherence condition) There exists a constant $\nu>0$ such that
		\begin{align*}
			\frac{Kn^2}{r}\max_{i\in[Kn^2]}\|\U^T\e_i\|_2^2\leq\nu~~\text{and}~~\frac{KN_p}{r}\max_{j\in[KN_p]}\|\V^T\e_j\|_{2}^2\leq\nu
		\end{align*}
		where $\U$ and $\V$ are left and right singular vectors of $\H^*$ defined as \cref{USigmaV:U} and \cref{USigmaV:V}.
	\end{enumerate}
\end{assump}

\begin{rmk} Before introducing our main results, we briefly discuss meanings of \cref{assumptions}, in particular the second to fourth assumptions.
	\begin{enumerate}
		\item[$\bullet$]The constant $4$ in Assumption \ref{assumptiononH} is related to the boundary condition used to define the two-fold Hankel matrix $\bmH(\z^*)$. If we discard patches involving the boundary extension of $\z^*$ or impose the periodic boundary condition, the quantities in Assumption \ref{assumptiononH} are bounded by $1$. On the other hand, if we impose the symmetric boundary condition, the quantities are bounded by $4$.
		
		\item[$\bullet$]Assumption \ref{pixeloccurrence} states that the occurrence of each pixel in the matrix $\H^*$ needs to be of the same order. This means that each pixel should appear a similar number of times in $\bmH(\z^*)$. If some pixels appear only a few times, we have very limited information about them, making it difficult to restore those pixels accurately.
		
		\item[$\bullet$]The incoherence assumption (Assumption \ref{incoherence}) is a common assumption in completion tasks \cite{wei2016guarantees}. It ensures that the energy of the matrices $\U$ and $\V$ is evenly distributed along their rows. This implies that there is some information sampled by each operator, avoiding the pessimistic case when the energy is highly concentrated in few entries.
	\end{enumerate}	
\end{rmk}

Under these assumptions, our main results are summarized as follows.

\begin{thm}\label{Theorem1} Assume that \cref{assumptions} holds. Let $\bOm=\left[N\right]^2$ and $\La=\left\{\k_1,\cdots,\k_m\right\}\subseteq\bOm$ be a data set whose entry is i.i.d. randomly drawn from the uniform distribution on $\bOm$ and assume $m\leq N^2$. Assume that
	\begin{align*}
		%\min\left\{\frac{n^2}{N^2},\frac{N_p}{N^2}\right\}\geq c_s,~~\text{and}~~\max\left\{\frac{n^2}{N^2},\frac{N_p}{N^2}\right\}\leq C_s
		c_s\leq\frac{n^2}{N^2},\frac{N_p}{N^2}\leq C_s
	\end{align*}
	for some absolute constants $c_s>0$ and $C_s>0$. Then there exists a constant $C>0$ depending only on $c_s$ such that if
	\begin{align*}
		m \geq CMK^{-1} \nu r\log^2(N),
	\end{align*}
	\begin{enumerate}
		\item $\z^*$ is the unique solution to \cref{prob:noisefree} with probability at least $1-4|\bOm|^{-4}$;
		\item $\left\|\z-\z^*\right\|_2\leq \wt{C}\left(MKN^2\max\{n^2,N_p\}\right)^{1/2}\delta$ with probability at least $1-4|\bOm|^{-4}$ for some absolute constant $\wt{C}>0$, where $\z$ is a solution to \cref{prob:noise} with $\left\|\bmP_{\La}\left(\z^*\right)-\y\right\|_2\leq\sqrt{m}\delta$.
	\end{enumerate}
\end{thm}

The proof is presented in Section \ref{sec:proofs}. Before concluding this section, we discuss the implications of Theorem \ref{Theorem1}.

Firstly, if we assume that $M>0$ in Assumption \ref{pixeloccurrence} is an absolute constant, Theorem \ref{Theorem1} asserts that under the incoherence condition, problem \cref{prob:noisefree} enables exact restoration with high probability as long as the number of samples $m$ exceeds $O(r\log^2(N))$.
If $K=1$, then problems \cref{prob:noisefree} and \cref{prob:noise} reduce to Hankel matrix completion through nuclear norm minimization. In other words, our problem encompasses spectral compressed sensing as a special case \cite{cai2016robust}, indicating that the degrees of freedom are at least of order $O(r)$.
When assuming the number of groups $K = O(1)$, the parameters needed to represent $\bmH(\z^*)$ is at least of the order $r$.
Thus, the number of samples matches the degrees of freedom up to logarithmic factors.

Furthermore, when the measurement $\y$ is contaminated by an error of size $O(\delta)$ on average, the distance between the original image and the solution to problem \cref{prob:noise} is bounded by the measurement error. In the lifted matrix, the entry-wise error is on average bounded by $O(\sqrt{\max\{n^2,N_p\}}\delta)$, which also serves as the average entry-wise error bound for the restored image.

\section{Proof of \cref{Theorem1}}\label{sec:proofs}
In this section we present the proof of \cref{Theorem1}.
All the lemmas stated in this section are under the setting of Theorem~\ref{Theorem1}.

We first introduce some notations which will be used throughout.
For simplicity, in the following we shall denote by $N_1 = Kn^2, N_2 = KN_p$.
From the way we construct $\H^*$ from $\z^*$ as in \eqref{DefineH*}, an operator $\bmG: \RR^{n\times n} \rightarrow\RR^{N_1\times N_2}$ is introduced as a shorthand notation as follows:
$$\bmG(\z) := \bdiag\left(\left\{\bmH\left(\z^*\right)\big|_{\bOm_k}:k=1,\cdots,K\right\}\right).$$
And we denote the range of $\bmG$ by $\range(\bmG) = \VV$ and $\dim(\VV) = N^2$.
For each $\oom = (\omega_1,\omega_2)\in \bOm$, we define $\B_{\oom}$ as follows
\begin{align}
	\B_{\oom} = \frac{1}{\sqrt{c_{\oom}}}\bmG(\e_{\omega_1}\e_{\omega_2}^T),
\end{align}
where $c_{\oom}:= \sum_{k=1}^K\left\|\bmH(\e_{\omega_1}\e_{\omega_2}^T)\big|_{\bOm_k}\right\|_{0}$ is the number of occurrence of pixel $\z^*_{\oom}$ in the matrix $\H^*$.
For ease of representation, we define the following operators $\bmB,\bmB_{\La},\bmB_{\oom},\bmB_{\La}',\bmB^{\perp}:\RR^{N_1\times N_2} \rightarrow \RR^{N_1\times N_2}$ for all $\oom\in\Omega$ as follows:
\begin{align}
	\bmB(\M) &= \sum_{\oom\in\bOm}\inp{\M}{\B_{\oom}}\cdot\B_{\oom}\label{def:operators_calB}\\
	\bmB_{\La}(\M) &= \sum_{\oom\in\La}\inp{\M}{\B_{\oom}}\cdot\B_{\oom}\label{def:operators_calB_Omega}\\
	\bmB_{\oom}(\M) &= \inp{\M}{\B_{\oom}}\B_{\oom}
\end{align}
and $\bmB_{\La}'$ is defined similarly to \eqref{def:operators_calB_Omega} but with the sum extending only over distinct samples. And $\bmB^{\perp}$ is the complement of $\bmB$, i.e. $\bmB^{\perp} = \bmI - \bmB$. Notice here $\bmB$ is the orthogonal projector onto $\range(\bmG)$ and $\ker(\bmB^{\perp}) = \range(\bmG)$. Also, we remind the readers that $\bmB_{\La}$ is not a projection operator but $\bmB_{\La}'$is since an element in $\bOm$ may appear several times in $\La$. We also define the following norms on the space $\RR^{N_1\times N_2}$ which will be useful throughout:
\begin{align}
	\bnorm{\M} &:= \left(\sum_{\oom\in\bOm} b_{\oom}^2\cdot|\inp{\M}{\B_{\oom}}|^2\right)^{1/2},\label{alt:norm:bnorm}\\
	\binfnorm{\M} &:= \max_{\oom\in\bOm} b_{\oom}\cdot|\inp{\M}{\B_{\oom}}|,\label{alt:norm:binf}
\end{align}
where $b_{\oom} = \op{\B_{\oom}}$ is a shorthand notation.
We denote the tangent space at $\H^*$ by
$$\TT:=\{\U\M_1^T + \M_2\V^T: \M_1\in\RR^{N_1\times r}, \M_2\in\RR^{N_2\times r}\}.$$
And we use $\bmP_{\TT}$ as the project operator onto the tangent space. Moreover, for any $\M\in\RR^{N_1\times N_2}$,
$\calP_{\TT}\M = \U\U^T\M + \M\V\V^T - \U\U^T\M\V\V^T. $

Since the operator $\bmG$ is injective, there is an one-to-one correspondence between the image domain $\RR^{n\times n}$ and $\range(\bmG)$. Together with the fact that $\bmG(\z^*) = \H^*$, we are able to write problem~\eqref{prob:noisefree} in an equivalent form as a nuclear norm minimization problem in the space $\RR^{N_1\times N_2}$:
\begin{align}\label{prob:nuclear_norm_lifted}
	\min\nuc{\H },~\text{s.t.}~\bmB^{\perp}(\H ) = 0~\text{and}~\bmB_{\Omega}(\H ) = \bmB_{\Omega}(\H ^*).
\end{align}
The first constraint in \eqref{prob:nuclear_norm_lifted} dictates that the lifted matrix $\H$ should be in $\ker(\bmB^{\perp}) = \range(\bmG)$ and the second constraint implies the information that we know about the original image $\z^*$. Now since $\bmG(\z^*) = \H^*$, we can prove the first statement of \cref{Theorem1} by showing $\H^*$ is the unique solution to problem~\eqref{prob:nuclear_norm_lifted}.

\hspace{1cm}

\noindent\textit{Some facts about the sampling basis $\{\B_{\oom}\}$.}
We begin with some observations about the sampling basis $\{\B_{\oom}\}$.
\begin{obs}[Facts about $\{\B_{\oom}\}$]\label{obs:Bomega}
	Let $\La_{\oom} := \supp(\B_{\oom})\subset [N_1]\times[N_2]$, i.e., the set of indices on which $\B_{\omega}$ is nonzero and we denote $c_{\oom} = |\La_{\oom}|$. Also, we denote the operator norm of $\B_{\oom}$ by $b_{\oom} = \op{\B_{\oom}}$. Then $\B_{\oom}$ satisfy the following:
	\begin{enumerate}
		\item For all $\oom\in\bOm$, each nonzero entry of $\B_{\oom}$ is $c_{\oom}^{-1/2}$.
		\item For all $\aal\neq\bbe\in\bOm$, we have $\La_{\aal}\cap\La_{\bbe} = \emptyset$. \item For all $\oom\in\bOm$, $\B_{\oom}$ and $\U,\V$ defined in \eqref{USigmaV:U} and \eqref{USigmaV:V},
		$$\supp(\calP_{\TT}(\B_{\oom})), \supp(\U\V^T)\subset\cup_{\aal\in\bOm}\La_{\aal}.$$
		And $|\cup_{\aal\in\bOm}\La_{\aal}| = Kn^2N_p = N_1N_2/K$.
		\item For all $\oom\in\bOm$, each row and column of $\B_{\oom}$ has at most $4$ nonzero entries.
	\end{enumerate}
\end{obs}
As a consequence of \cref{assumptions}, we have
\begin{lem}\label{lemma:opnorm:B:general}
	For all $\oom\in \bOm$, we have
	$\frac{1}{\sqrt{c_{\oom}}} \leq b_{\oom} \leq \frac{4}{\sqrt{c_{\oom}}}.$
	And
	$$\min_{\oom}c_{\oom}\geq  \frac{KN_pn^2}{MN^2},\qquad
	\max_{\oom}b_{\oom}\leq 4\sqrt\frac{MN^2}{KN_pn^2}.$$
	Moreover, we have
	$\max_{\oom_1,\oom_2} \frac{b_{\oom_1}}{b_{\oom_2}} \leq 4\sqrt{M}.$
\end{lem}
\begin{proof}
	For any $\oom\in\bOm$, since $M<+\infty$ in \cref{assumptions}, there exists $i_{\oom}\in[N_2]$, $\B_{\oom}\e_{i_{\oom}}\neq 0$. Also, since the non-zero entries of $\B_{\oom}$ are all $c_{\oom}^{-1/2}$, we see that
	$$b_{\oom} = \op{\B_{\oom}} = \max_{\ltwo{\x}\leq 1}\ltwo{\B_{\oom}\x}\geq \ltwo{\B_{\oom}\e_{i_{\oom}}}\geq c_{\oom}^{-1/2}.$$	
	On the other hand, we use Holder's inequality, and we get
	$$b_{\oom} = \op{\B_{\oom}} \leq \sqrt{\|\B_{\oom}\|_{1} \cdot \|\B_{\oom}\|_{\infty}}\leq 4c_{\oom}^{-1/2},$$
	where $\|\A\|_{1}:= \max_{j\in[N_2]}\sum_{i\in[N_1]}|\A_{ij}|$ and $\|\A\|_{\infty}:= \max_{i\in[N_1]}\sum_{j\in[N_2]}|\A_{ij}|$.
	%	
	%The lower bound for $b_{\oom}$ is clear since $\B_{\oom}$ is the matrix containing entries with value only $0$ and $\frac{1}{\sqrt{c_{\oom}}}$. For the upper bound for $b_{\oom}$, from the definition of $\B_{\omega}$, each column and row of $\B_{\omega}$ contains at most $4$ non-zero elements with magnitude $1/\sqrt{c_{\omega}}$. Then as a result of Holder's inequality, we have
	%$$b_{\omega} \leq \sqrt{\|\B_{\omega}\|_{1} \cdot \|\B_{\omega}\|_{\infty}} \leq \frac{4}{\sqrt{c_{\omega}}}.$$
	Finally from the third fact in Observation~\ref{obs:Bomega}, we have
	$$KN_pn^2 = \sum_{\oom\in\bOm}c_{\oom} \leq |\bOm| \max_{\oom\in\bOm}c_{\oom}\leq M|\bOm|\min_{\oom\in\bOm}c_{\oom},$$ so we have
	$\min_{\oom}c_{\oom}\geq \frac{KN_pn^2}{MN^2}$. Meanwhile,
	$$\max_{\oom}b_{\oom} \leq \max_{\oom}4c_{\oom}^{-1/2}\leq 4\sqrt\frac{MN^2}{KN_pn^2}.$$
	And
	$$\max_{\oom_,\oom_2}\frac{b_{\oom_1}}{b_{\oom_2}}\leq 4\max_{\oom_1,\oom_2}\sqrt\frac{c_{\oom_1}}{c_{\oom_2}}\leq 4\sqrt{M}.$$
	And this finishes the proof of the lemma.
\end{proof}
%\begin{rmk}
%	Under \cref{assumptions} and assume $K$ is bounded by some absolute constant, then $\mb,\mnorm$ are also an absolute constants.
%\end{rmk}
And the following lemmas estimate $\fro{\calP_{\TT}(\B_{\oom})}^2$ and $\bnorm{\calP_{\TT}(b_{\oom}^{-1}\B_{\oom})}^2$.
\begin{lem}\label{lemma:incoh}
	%	Under the setting of Theorem~\ref{thm:noiseless},
	For all $\oom\in\bOm$, we have
	$$\fro{\calP_{\TT}(\B_{\omega})}^2 \leq 8c_s^{-2}MK^{-1}\nu\frac{r}{N^2},$$
	and
	$$
	\bnorm{\calP_{\TT}(b_{\oom}^{-1}\B_{\oom})}^2 \leq 128c_s^{-2}M^2K^{-1}\nu\frac{r}{N^2}.
	$$
\end{lem}
\begin{proof}
	Recall $\calP_{\TT}\M = \U\U^T\M + \M\V\V^T - \U\U^T\M\V\V^T$, we have
	\begin{align}\label{bound:ptb}
		\fro{\calP_{\TT}(\B_{\oom})}^2 &\leq \fro{\U^T\B_{\oom}}^2 + \fro{\V^T\B_{\oom}^T}^2 \overset{(a)}{\leq} \left(4N_2\frac{\nu r}{N_1} + 4N_1\frac{\nu r}{N_2}\right)\max_{\omega}c_{\omega}^{-1}\notag\\
		&\overset{(b)}{\leq}4KM\left(\frac{\nu r N^2}{N_1^2} + \frac{\nu rN^2}{N_2^2}\right)\notag\\
		&\overset{(c)}{\leq}8c_s^{-2}MK^{-1}\nu\frac{r}{N^2},
	\end{align}
	where in $(a)$ we use the last fact in Observation~\ref{obs:Bomega}, in $(b)$ we use Lemma~\ref{lemma:opnorm:B:general} and $(c)$ follows from $c_s\leq \frac{n^2}{N^2},\frac{N_p}{N^2}$.
	
	On the other hand, from the definition of $\bnorm{\cdot}$, we have
	\begin{align}
		\bnorm{\calP_{\TT}(b_{\oom}^{-1}\B_{\oom})}^2 &= \sum_{\aal\in\bOm}|\inp{\calP_{\TT}(b_{\oom}^{-1}\B_{\oom})}{b_{\aal}\B_{\aal}}|^2\\
		&\overset{(a)}{=} \sum_{\aal\in\bOm}\left(\frac{b_{\aal}}{b_{\oom}}\right)^2 c_{\aal}^{-1}\left(\sum_{\gga\in\La_{\aal}}[\bmP_{\TT}(\B_{\oom})]_{\gga}\right)^2\no\\
		&\overset{(b)}{\leq} \sum_{\aal\in\bOm}\left(\frac{b_{\aal}}{b_{\oom}}\right)^2 \sum_{\gga\in\La_{\aal}}[\bmP_{\TT}(\B_{\oom})]_{\gga}^2\no\\
		&\overset{(c)}{\leq} 16M\sum_{\aal\in\bOm}\sum_{\gga\in\La_{\aal}}[\bmP_{\TT}(\B_{\oom})]_{\gga}^2\no\\
		&\overset{(d)}{=}16M\fro{\bmP_{\TT}(\B_{\oom})}^2\overset{(e)}{\leq} 128c_s^{-2}M^2K^{-1}\nu\frac{r}{N^2},
	\end{align}
	where $(a)$ is from the first fact of Observation~\ref{obs:Bomega}; in $(b)$ we use AM-GM inequality; $(c)$ follows from Lemma~\ref{lemma:opnorm:B:general}; and $(d)$ follows from the third fact of Observation~\ref{obs:Bomega}; $(e)$ is from \eqref{bound:ptb}.
\end{proof}
For the sake of simplicity, in the following, we denote $\mu:= 8c_s^{-2}MK^{-1}\nu$.

\subsection{Exact recovery}\label{sec:exact}

\begin{lem}\label{lemma:uniqueness:general}
	Let $\La=\left\{\k_1,\cdots,\k_m\right\}\subseteq\bOm$ contain i.i.d. samples drawn from the uniform distribution of $\bOm$. Suppose that the sampling operator $\bmB_{\La}$ satisfies:
	\begin{align}\label{eq:lemma:uniqueness:cond1:general}
		\|\frac{N^2}{m}\bmP_{\TT}\bmB_{\bOm}\bmP_{\TT}-\bmP_{\TT}\bmB\bmP_{\TT}\|\leq\frac{1}{2},
	\end{align}
	and there exists a matrix $\Y$ satisfying:
	\begin{align}
		(\bmB-\bmB_{\La}') (\Y) &= 0,\label{dualcond:1:general}\\
		\|\bmP_{\TT}^{\perp}(\Y)\|&\leq\frac{1}{2},\label{dualcond:2:general}\\
		\fro{\U\V^T - \bmP_{\TT}(\Y)}&<\frac{1}{2N^4},\label{dualcond:3:general}
	\end{align}
	where $\U,\V$ are left and right singular vectors of $\H^*$, i.e. $\H^* = \U\bSigma\V^T$ is the compact SVD.
	Then $\H^*$ is the unique solution to problem~\eqref{prob:nuclear_norm_lifted}.
\end{lem}
\begin{proof}
	Consider any non-zero perturbation matrix $\Z$ that satisfies $\bmB^{\perp}(\Z) = \bzero$ and $\bmB_{\La}(\Z) = \bzero$. Notice $\bmB_{\La}(\Z) = \bzero$ is equivalent to $\bmB'_{\La}(\Z) = \bzero$.
	
	Now we consider the subgradient of the function $\nuc{\H}$ at the point $\H^* = \U\bSigma\V^T$. Let the compact SVD of $\bmP_{\TT}^{\perp}(\Z)$ be $\bmP_{\TT}^{\perp}(\Z) = \L\S\R^T$, and we set
	\begin{align}\label{eq:W:general}
		\W = \U\V^T + \L\R^T,
	\end{align}
	then $\W$ is indeed a subgradient since $\op{\L\R^T} \leq 1$ and $\L^T\U = \bzero$, $\R^T \V = \bzero$.
	%	Also as a result, we have $\inp{\bmP^{\perp}_{\TT}(\Z)}{\W} = \nuc{\bmP^{\perp}_{\TT}(\Z)}$.
	
	Now we show that $\H^*$ is indeed the unique minimizer. Consider the nuclear norm of $\H^* + \Z$, then from the definition of subgradient, we have,
	\begin{align}\label{eq:est:nuc:H+Z:general}
		\nuc{\H^*+\Z} &\geq \nuc{\H^*} + \inp{\Z}{\W} \overset{(a)}{=} \nuc{\H^*} + \inp{\Z}{\W-\Y}\no\\
		&= \nuc{\H^*} + \inp{\bmP_{\TT}(\Z)}{\bmP_{\TT}(\W-\Y)} + \inp{\bmP_{\TT}^{\perp}(\Z)}{\bmP_{\TT}^{\perp}(\W-\Y)},
	\end{align}
	where $(a)$ holds since $\inp{\Z}{\Y} = \inp{(\bmB-\bmB_{\La}')(\Z)}{\Y}  = \inp{\Z}{(\bmB-\bmB_{\La}')(\Y)}= 0$.
	Now the term $\inp{\bmP_{\TT}^{\perp}(\Z)}{\bmP_{\TT}^{\perp}(\W-\Y)}$ is lower bounded as follows,
	\begin{align}\label{eq:est:inp:ZW:1:general}
		\inp{\bmP_{\TT}^{\perp}(\Z)}{\bmP_{\TT}^{\perp}(\W-\Y)} &= \inp{\bmP_{\TT}^{\perp}(\Z)}{\bmP_{\TT}^{\perp}(\W)} - \inp{\bmP_{\TT}^{\perp}(\Z)}{\bmP_{\TT}^{\perp}(\Y)} \no\\
		&\overset{(a)}{=}\nuc{\bmP_{\TT}^{\perp}(\Z)} - \op{\bmP_{\TT}^{\perp}(\Y)}\cdot\nuc{\bmP_{\TT}^{\perp}(\Z)}\overset{(b)}{\geq} \frac{1}{2}\nuc{\bmP_{\TT}^{\perp}(\Z)},
	\end{align}
	where $(a)$ holds from the choice of $\W$ as in \eqref{eq:W:general} and we also use the inequality that $\inp{\A}{\B} \leq \op{\A}\nuc{\B}$; and $(b)$ holds from \eqref{dualcond:2:general}. Now we plug \eqref{eq:est:inp:ZW:1:general} into \eqref{eq:est:nuc:H+Z:general} and use Cauchy-Schwartz inequality, we get,
	\begin{align}\label{eq:est:H+Z:1:general}
		\nuc{\H^*+\Z}&\geq \nuc{\H^*}-\fro{\bmP_{\TT}(\Z)}\cdot\fro{\bmP_{\TT}(\W-\Y)} + \frac{1}{2}\nuc{\bmP_{\TT}^{\perp}(\Z)}\no\\
		&\overset{(a)}{\geq} \nuc{\H^*}-\frac{1}{2N^4}\fro{\bmP_{\TT}(\Z)}+ \frac{1}{2}\nuc{\bmP_{\TT}^{\perp}(\Z)},
	\end{align}
	where in $(a)$ we use \eqref{dualcond:3:general}.
	We consider two different cases.
	
	\noindent\textit{Case 1: $\fro{\bmP_{\TT}^{\perp}(\Z)} > \frac{1}{N^4}\fro{\bmP_{\TT}(\Z)}$.}
	From \eqref{eq:est:H+Z:1:general}, and that $\nuc{\bmP_{\TT}^{\perp}(\Z)}\geq \fro{\bmP_{\TT}^{\perp}(\Z)}$, we have
	\begin{align}
		\nuc{\H^*+\Z}&\geq \nuc{\H^*}-\frac{1}{2N^4}\fro{\bmP_{\TT}(\Z)} + \frac{1}{2}\fro{\bmP_{\TT}^{\perp}(\Z)}\no\\
		&> \nuc{\H^*}-\frac{1}{2N^4}\fro{\bmP_{\TT}(\Z)} + \frac{1}{2N^4}\fro{\bmP_{\TT}(\Z)}=\nuc{\H^*}.
	\end{align}
	
	\noindent\textit{Case 2: $\fro{\bmP_{\TT}^{\perp}(\Z)} \leq \frac{1}{N^4}\fro{\bmP_{\TT}(\Z)}$.}
	We will show under this case, $\Z = \bzero$. First we derive a lower bound for $\fro{\left(\frac{N^2}{m}\bmB_{\La} + \bmB^{\perp}\right)\bmP_{\TT}(\Z)}^2$.
	\begin{align}\label{eq:ptz:general}
		\fro{\left(\frac{N^2}{m}\bmB_{\La} + \bmB^{\perp}\right)\bmP_{\TT}(\Z)}^2 &= \frac{N^4}{m^2}\inp{\bmB_{\La}\bmP_{\TT}(\Z)}{\bmB_{\La}\bmP_{\TT}(\Z)} + \inp{\bmB^{\perp}\bmP_{\TT}(\Z)}{\bmB^{\perp}\bmP_{\TT}(\Z)}\no\\
		&\overset{(a)}{\geq} \inp{\bmP_{\TT}(\Z)}{\frac{N^2}{m}\bmB_{\La}\bmP_{\TT}(\Z)} + \inp{\bmP_{\TT}(\Z)}{\bmB^{\perp}\bmP_{\TT}(\Z)}\no\\
		&= \inp{\bmP_{\TT}(\Z)}{\bmP_{\TT}(\frac{N^2}{m}\bmB_{\La}+\bmB^{\perp})\bmP_{\TT}(\Z)}\no\\
		&= \fro{\bmP_{\TT}(\Z)}^2 + \inp{\bmP_{\TT}(\Z)}{\bmP_{\TT}(\frac{N^2}{m}\bmB_{\La}-\bmB)\bmP_{\TT}(\Z)}\no\\
		&\overset{(b)}{\geq} \frac{1}{2}\fro{\bmP_{\TT}(\Z)}^2,
	\end{align}
	where $(a)$ is from the fact that $\inp{\bmB_{\La}(\H)}{\bmB_{\La}(\H)} \geq \inp{\H}{\bmB_{\La}(\H)}$ and that $N^2\geq m$ and $(b)$ is from \eqref{eq:lemma:uniqueness:cond1:general}. Meanwhile,
	\begin{align}\label{eq:ptperpz:general}
		\fro{\left(\frac{N^2}{m}\bmB_{\La} + \bmB^{\perp}\right)\bmP_{\TT}^{\perp}(\Z)} \leq \op{\frac{N^2}{m}\bmB_{\La} + \bmB^{\perp}}\fro{\bmP_{\TT}^{\perp}(\Z)}\overset{(a)}{\leq}N^2\fro{\bmP_{\TT}^{\perp}(\Z)},
	\end{align}
	where $(a)$ holds since $\op{\frac{N^2}{m}\bmB_{\La} + \bmB^{\perp}} \leq \frac{N^2}{m}\left(\op{\bmB_{\k_1}+\bmB^{\perp}}+\sum_{i=2}^m\op{\bmB_{\k_i}}\right)\leq n_{\VV}$. Since $\bmB^{\perp}\Z = \bmB_{\La}\Z = \bzero$, we have $\left(\frac{N^2}{m}\bmB_{\La} + \bmB^{\perp}\right)\Z= \bzero$. As a result of this observation and \eqref{eq:ptz:general} and \eqref{eq:ptperpz:general}, we have
	\begin{align}
		0 &= \fro{\left(\frac{N^2}{m}\bmB_{\La} + \bmB^{\perp}\right)\Z}\geq \fro{\left(\frac{N^2}{m}\bmB_{\La} + \bmB^{\perp}\right)\bmP_{\TT}(\Z)} -
		\fro{\left(\frac{N^2}{m}\bmB_{\La} + \bmB^{\perp}\right)\bmP_{\TT}^{\perp}(\Z)}\no\\
		&\geq \frac{\sqrt{2}}{2}\fro{\bmP_{\TT}(\Z)} - N^2\fro{\bmP_{\TT}^{\perp}(\Z)}\no\\
		&\overset{(a)}{\geq}\frac{\sqrt{2}}{2}\fro{\bmP_{\TT}(\Z)} - \frac{1}{N^2}\fro{\bmP_{\TT}(\Z)},
	\end{align}
	where $(a)$ holds since $\fro{\bmP_{\TT}^{\perp}(\Z)} \leq \frac{1}{N^4}\fro{\bmP_{\TT}(\Z)}$. Since $\frac{1}{\sqrt{2}} - \frac{1}{N^2} > 0$ so we must have $\bmP_{\TT}(\Z) = \bmP_{\TT}^{\perp}(\Z) = \bzero$ and thus $\Z = \bzero$.
	
	Combine these two cases and we claim that $\H^*$ is the unique minimizer, which finishes the proof of the lemma.
\end{proof}

Notice that the condition \eqref{eq:lemma:uniqueness:cond1:general} is satisfied with high probability as guaranteed by the following lemma.

\begin{lem}\label{lemma:estofptbomegapt-expectation}
	For any small constant $0<\epsilon \leq \frac{1}{2}$, we have
	$$
	\|\frac{N^2}{m}\bmP_{\TT}\bmB_{\La}\bmP_{\TT}-\bmP_{\TT}\bmB\bmP_{\TT}\| \leq \epsilon,
	$$
	with probability exceeding $1-(N_1+N_2)^{-10}$, provided that $m\geq C_{\epsilon}\cdot \mu r\log(N_1+N_2)$ for some constant $C_{\epsilon} > 0$ depending only on $\epsilon$.
\end{lem}
\begin{proof}
	For $\forall \oom\in\bOm$, denote
	$$\bmZ_{\oom} := \frac{N^2}{m}\bmP_{\TT}\bmB_{\omega}\bmP_{\TT} - \frac{1}{m}\bmP_{\TT}\bmB\bmP_{\TT}.$$
	Then for arbitrary matrix $\M$, we have,
	\begin{align}\label{eq:ptbwptsquare:general}
		(\bmP_{\TT}\bmB_{\oom}\bmP_{\TT})^2(\M) = \inp{\bmP_{\TT}(\B_{\oom})}{\B_{\oom}}\inp{\bmP_{\TT}(\B_{\oom})}{\M}\cdot\bmP_{\TT}(\B_{\oom}) = \fro{\bmP_{\TT}(\B_{\oom})}^2\bmP_{\TT}\bmB_{\oom}\bmP_{\TT}(\M).
	\end{align}
	Notice since $\bmP_{\TT}\bmB_{\oom}\bmP_{\TT}$ is a symmetric positive semi-definite operator, and together with \eqref{eq:ptbwptsquare:general} and Lemma~\ref{lemma:incoh}, we get,
	\begin{align}\label{eq:est:ptbwpt:general}
		\op{\bmP_{\TT}\bmB_{\oom}\bmP_{\TT}} \leq \fro{\bmP_{\TT}(\B_{\oom})}^2\leq \mu\frac{r}{N^2}.
	\end{align}
	Now let $\k_i\in\La$ be an index drawn uniformly from $\bOm$. Then we have from \eqref{eq:est:ptbwpt:general},
	\begin{align}
		\op{\bmZ_{\k_i}} = \op{\frac{N^2}{m}\bmP_{\TT}\bmB_{\k_i}\bmP_{\TT} - \frac{1}{m}\bmP_{\TT}\bmB\bmP_{\TT}}\leq 2\mu\frac{r}{m}.
	\end{align}
	On the other hand, we have $\EE \bmZ_{\k_i}^2 = \frac{N^4}{m^2}\EE(\bmP_{\TT}\bmB_{\k_i}\bmP_{\TT})^2 - \frac{1}{m^2}(\bmP_{\TT}\bmB\bmP_{\TT})^2\leq \frac{N^4}{m^2}\EE(\bmP_{\TT}\bmB_{\k_i}\bmP_{\TT})^2$ and thus,
	\begin{align}
		\op{\sum_{i=1}^m\EE\bmZ_{\omega_i}^2} &\leq \frac{N^4}{m^2}\op{\sum_{i=1}^m\EE (\bmP_{\TT}\bmB_{\k_i}\bmP_{\TT})^2}\leq\frac{N^4}{m}\op{\EE (\bmP_{\TT}\bmB_{\k}\bmP_{\TT})^2}\notag\\
		&\overset{(a)}{\leq} \mu\frac{r}{N^2}\frac{N^4}{m}\op{\EE \bmP_{\TT}\bmB_{\k}\bmP_{\TT}}
		\leq \mu\frac{r}{m},
	\end{align}
	where in $(a)$ we use \eqref{eq:est:ptbwpt:general}. Applying Bernstein concentration inequality and we get with probability exceeding $1-(N_1+N_2)^{-10}$,
	$$\op{\sum_{i = 1}^m \bmZ_{\k_i}} \leq C\max\left\{\left(\frac{\mu r}{m}\log(N_1+N_2)\right)^{1/2}, \frac{\mu r}{m}\log(N_1+N_2)\right\},$$
	for some absolute constant $C>0$. And this finished the proof of the lemma. 	
\end{proof}

\hspace{1cm}

\noindent\textit{Dual certificate.}
This section is for constructing the dual certificate as required in \cref{lemma:uniqueness:general}.
Now we are ready to construct the dual certificate $\Y$ that satisfies \eqref{dualcond:1:general}-\eqref{dualcond:3:general}. And it is constructed using the golfing scheme introduced in \cite{gross2011recovering}. Suppose that we have $L$ independent random location sets $\La_i, (1\leq i\leq L)$, each containing $\frac{m}{L}$ i.i.d. samples.

\begin{algorithm}[H]\label{algo:golfingscheme:general}
	\caption{Construction of dual certificate $\Y$ using golfing scheme}
	\begin{algorithmic}
		\STATE{Set $\F_0 = \U\V^T$, $L = \lceil 4\log N\rceil$ and $q = \frac{m}{N^2L}$.}
		\FOR{$i=1$ to $L$}
		\STATE{$\F_i = \bmP_{\TT}(\bmB-\frac{1}{q}\bmB_{\La_i})\bmP_{\TT}(\F_{i-1})$.}
		\ENDFOR
		\STATE{$\Y = \sum_{i = 1}^{L} (\frac{1}{q}\bmB_{\La_i} +\bmB^{\perp})(\F_{i-1})$.}	
	\end{algorithmic}
\end{algorithm}
Now we verify that the $\Y$ constructed in this way indeed satisfies the three conditions.

\hspace{1cm}

\noindent\textit{Verification of $(\bmB-\bmB_{\La}') (\Y) = 0$.} Notice that for all $1\leq i\leq L$, we have for all $\oom\in\bOm$,
$(\bmB-\bmB_{\La}')\bmB_{\oom} = \bzero$ and $(\bmB-\bmB_{\La}')\bmB^{\perp} = \bzero$ and thus $(\bmB-\bmB_{\La}') (\Y) = 0$ from the definition of $\Y$.

\hspace{1cm}

\noindent\textit{Estimation of $\op{\bmP_{\TT}^{\perp}(\Y)}$.} From the construction of $\Y$, we have
\begin{align}\label{eq:est:ptperpY:general}
	\op{\bmP_{\TT}^{\perp}(\Y)} &= \left\|\bmP_{\TT}^{\perp}\left(\sum_{i=1}^L (\frac{1}{q}\bmB_{\La_i} + \bmB^{\perp})(\F_{i-1})\right)\right\|\leq \sum_{i=1}^L \op{\bmP_{\TT}^{\perp}(\frac{1}{q}\bmB_{\La_i} + \bmB^{\perp})(\F_{i-1})}\no\\
	&\overset{(a)}{=} \sum_{i=1}^L \op{\bmP_{\TT}^{\perp}(\frac{1}{q}\bmB_{\La_i} - \bmB)(\F_{i-1})}\no\\
	&\leq  \sum_{i=1}^L\op{(\frac{1}{q}\bmB_{\La_i} - \bmB)(\F_{i-1})},
\end{align}
where $(a)$ holds since $\F_{i-1}$ lies in the tangent space $\TT$ by construction. Now the following lemma gives the estimation for the summand under the settings of \cref{Theorem1}.

\begin{lem}[Estimation of $\|(\frac{N^2}{m_0}\bmB_{\La_0} -\bmB)(\M)\|$]\label{lemma:estofbm:general}
	Let $\La_0$ be an index set such that $|\La_0| = m_0$. Then we have for any given matrix $\M$, there exists an absolute constant $C>0$, such that
	$$
	\op{(\frac{N^2}{m_0}\bmB_{\La_0} -\bmB)(\M)} \leq C\left(\sqrt{\frac{N^2\log(N_1+N_2)}{m_0}}\bnorm{\M} + \frac{N^2\log(N_1+N_2)}{m_0}\binfnorm{\M}\right)
	$$
	holds with probability at least $1-(N_1+N_2)^{-10}$.
\end{lem}
\begin{proof}
	Set $\La_0 = \{\l_1,\ldots,\l_{m_0}\}$.
	Denote for all $\oom\in\bOm$,
	$$\S_{\oom} = \frac{N^2}{m_0}\bmB_{\oom}(\M) - \frac{1}{m_0}\bmB(\M).$$
	As a consequence, $(\frac{N^2}{m_0}\bmB_{\La_0} - \bmB)(\M) = \sum_{i=1}^{m_0}\S_{\l_i}$. We first estimate $\op{\S_{\l_i}}$.
	\begin{align*}
		\op{\S_{\l_i}} \leq \frac{2N^2}{m_0}\max_{\bbe\in\bOm}\op{\bmB_{\bbe}(\M)} = \frac{2N^2}{m_0}\max_{\bbe\in\bOm}|\inp{\M}{\B_{\bbe}}| \cdot\op{\B_{\bbe}}\leq \frac{2N^2}{m_0}\binfnorm{\M}.
	\end{align*}
	On the other hand, we have
	\begin{align}\label{eq:est:S:general}
		\sum_{\oom\in\bOm}\S_{\oom}\S_{\oom}^T &= \sum_{\oom}\left(\frac{N^2}{m_0}\bmB_{\oom}(\M) - \frac{1}{m_0}\bmB(\M)\right) \left(\frac{N^2}{m_0}\bmB_{\omega}(\M) - \frac{1}{m_0}\bmB(\M)\right)^T\no\\
		&= \left(\frac{N^2}{m_0}\right)^2\sum_{\oom}\bmB_{\oom}(\M)\bmB_{\oom}(\M)^T-\frac{N^2}{m_0^2}\bmB(\M)\bmB(\M)^T\no\\
		&\overset{(a)}{\leq} \left(\frac{N^2}{m_0}\right)^2\sum_{\oom}\bmB_{\oom}(\M)\bmB_{\oom}(\M)^T,
	\end{align}
	here the relation $\A\leq\B$ in $(a)$ means $\B-\A$ is an SPSD matrix. So we conclude from \eqref{eq:est:S:general},
	\begin{align*}
		\op{\EE\sum_{i=1}^{m_0} \S_{\l_i}\S_{\l_i}^T} = \frac{m_0}{N^2}\op{\sum_{\oom\in\bOm}\S_{\oom}\S_{\oom}^T}\leq \frac{N^2}{m_0}\op{\sum_{\oom}\bmB_{\oom}(\M)\bmB_{\oom}(\M)^T} \leq \frac{N^2}{m_0}\bnorm{\M}^2.
	\end{align*}
	The same upper bound can be derived for $\op{\EE\sum_{i=1}^{m_0} \S_{\l_i}^T\S_{\l_i}}$. Since the size of the matrix is $N_1\times N_2$, applying Bernstein concentration inequality and we get with probability exceeding $1-(N_1+N_2)^{-10}$,
	$$\op{\sum_{i=1}^{m_0}\S_{\l_i}} \leq C\left(\sqrt{\frac{N^2\log(N_1+N_2)}{m_0}}\cdot\bnorm{\M} + \frac{N^2\log(N_1+N_2)}{m_0}\binfnorm{\M}\right),$$
	for some absolute constant $C>0$.
\end{proof}

\begin{lem}[Estimation of $\bnorm{\bmP_{\TT}(\frac{N^2}{m_0}\bmB_{\La_0} - \bmB)(\M)}$]\label{lemtwo:general}
	Let $\La_0$ be an index set such that $|\La_0| = m_0$. Then we have for any given matrix $\M$, there exists an absolute constant $C>0$, such that
	\begin{align*}
		&\bnorm{\bmP_{\TT}(\frac{n_{\VV}}{m_0}\bmB_{\La_0} - \bmB)(\M)} \leq C\left(\sqrt{\frac{1}{m_0}M\mu r\log N}\bnorm{\M} +  \frac{N^2}{m_0}\sqrt{M\mu\frac{r}{N^2}}\log N  \binfnorm{\M}\right)
	\end{align*}
	holds with probability at least $1-N^{-20}$.
\end{lem}
\begin{proof}
	Denote $\La_0 = \{\l_1,\ldots,\l_{m_0}\}$.
	For any given matrix $\M$, we have $$\bnorm{\bmP_{\TT}(\frac{N^2}{m_0}\bmB_{\La_0} - \bmB)(\M)}^2 = \sum_{\aal}b_{\aal}^2\cdot|\inp{\bmP_{\TT}(\frac{N^2}{m_0}\bmB_{\La_0} - \bmB)(\M)}{\B_{\aal}}|^2,$$
	recall $b_{\aal} = \op{\B_{\aal}}$.
	We first estimate $|\inp{\bmP_{\TT}(\frac{N^2}{m_0}\bmB_{\La_0} - \bmB)(\M)}{\B_{\aal}}|$. Denote for any $\oom,\aal\in\bOm$
	\begin{align}\label{eq:def:z:general}
		z_{\oom,\aal} := \inp{\bmP_{\TT}(\frac{N^2}{m_0}\bmB_{\oom} - \frac{1}{m_0}\bmB)(\M)}{b_{\aal}\B_{\aal}}.
	\end{align}
	Then as a result, we have
	$$|\inp{\bmP_{\TT}(\frac{N^2}{m_0}\bmB_{\La_0} - \bmB)(\M)}{b_{\aal}\B_{\aal}}| =
	\left|\sum_{i=1}^{m_0} z_{\l_i,\aal}\right|.$$
	Now we view $\z_{\oom} = \vec\left([z_{\oom,\aal}]_{\aal\in\bOm}\right)$ as a column vector of size $N^2$, then it is easy to see that
	$$\|\sum_{i=1}^{m_0}\z_{\l_i}\|_{2} = \bnorm{\bmP_{\TT}(\frac{N^2}{m_0}\bmB_{\La_0} - \bmB)(\M)}.$$
	We now investigate the upper bound for $\|\z_{\oom}\|_{2}$.
	\begin{align}\label{eq:est:0000:general}
		\|\z_{\oom}\|_{2}
		&\leq 2\sqrt{\sum_{\aal\in\bOm}|\inp{\bmP_{\TT}\frac{N^2}{m_0}\bmB_{\oom}(\M)}{b_{\aal}\B_{\aal}}|^2}\no\\
		&= \frac{2N^2}{m_0}\sqrt{\sum_{\aal}\inp{\bmP_{\TT}(b_{\oom}^{-1}\B_{\oom})}{b_{\aal}\B_{\aal}}^2\inp{\M}{\B_{\oom}}^2b_{\oom}^2}\no\\
		&\leq \frac{2N^2}{m_0}\bnorm{\bmP_{\TT}(b_{\omega}^{-1}\B_{\omega})}\cdot\binfnorm{\M}\no\\
		&\overset{(a)}{\leq} \frac{2N^2}{m_0}\sqrt{16M\mu\frac{r}{N^2}}\binfnorm{\M},
	\end{align}
	where in $(a)$ we use the \cref{lemma:incoh}.
	Meanwhile, we have $\EE(\sum_{i=1}^{m_0} \|\z_{\l_i}\|_{2}^2) = \frac{m_0}{N^2}\sum_{\oom}\|\z_{\oom}\|_{2}^2.$
	Meanwhile,
	\begin{align}
		\sum_{\oom}\|\z_{\oom}\|_{2}^2 &= \sum_{\oom,\bbe}\inp{b_{\bbe}\B_{\bbe}}{\frac{N^2}{m_0}\bmP_{\TT}\bmB_{\oom}(\M)}^2 - \inp{b_{\bbe}\B_{\bbe}}{\frac{1}{m_0}\bmP_{\TT}\bmB(\M)}^2\no\\
		&\leq \sum_{\oom,\bbe}\inp{b_{\bbe}\B_{\bbe}}{\frac{N^2}{m_0}\bmP_{\TT}\bmB_{\oom}(\M)}^2\no\\
		&= \frac{N^4}{m_0^2}\sum_{\oom,\bbe} \inp{b_{\bbe}\B_{\bbe}}{\bmP_{\TT}(b_{\oom}^{-1}\B_{\oom})}^2\cdot\inp{\M}{b_{\oom}\B_{\oom}}^2\no\\
		&\leq 16 \frac{N^4}{m_0^2}M\mu\frac{r}{N^2}\sum_{\oom}\inp{\M}{b_{\oom}\B_{\oom}}^2\no\\
		&=16\frac{N^2}{m_0^2}M\mu r\bnorm{\M}^2,
	\end{align}
	where the last inequality holds since $\max_{\oom,\bbe}\frac{b_{\bbe}}{b_{\oom}}\leq 4\sqrt{M}$ and $|\inp{\B_{\bbe}}{\bmP_{\TT}\B_{\oom}}|\leq \mu\frac{r}{N^2}$ from \cref{lemma:incoh}.
	This gives the bound for $\EE(\sum_{i=1}^{m_0} \|\z_{\l_i}\|_{\ell_2}^2)$ as follows,
	\begin{align}\label{eq:est:0001:general}
		\EE(\sum_{i=1}^{m_0} \|\z_{\l_i}\|_{\ell_2}^2) \leq 16\frac{1}{m_0}\mu r\bnorm{\M}^2.
	\end{align}
	On the other hand, since $\z_{\bomega}$ are vectors, we have,
	\begin{align}\label{eq:est:0002:general}
		\op{\EE \sum_{i=1}^{m_0} \z_{\l_i}\z_{\l_i}^T} \leq \EE\sum_{i=1}^{m_0}\|\z_{\l_i}\|_{2}^2.
	\end{align}
	Using Bernstein concentration inequality and we conclude with probability exceeding $1-N^{-20}$,
	$$\|\sum_{i=1}^{m_0}\z_{\l_i}\|_{\ell_2} \leq C\left(\sqrt{\frac{1}{m_0}M\mu r\log(N^2)}\bnorm{\M} +  \frac{N^2}{m_0}\sqrt{M\mu\frac{r}{N^2}}\log(N^2)  \binfnorm{\M}\right),$$
	for some absolute constant $C>0$.
\end{proof}

\begin{lem}[Estimation of $\binfnorm{\bmP_{\TT}(\frac{ N^2}{m_0}\bmB_{\La_0} - \bmB)(\M)}$]\label{lemthree:general}
	Let $\La_0$ be an index set such that $|\La_0| = m_0$. Then for any given matrix
	%		 $\red{\M\in\TT}$
	$\M$
	, there exists an absolute constant $C>0$, such that
	\begin{align*}
		&\binfnorm{\bmP_{\TT}(\frac{ N^2}{m_0}\bmB_{\La_0} - \bmB)(\M)}
		\leq C\bigg(\frac{M\mu r\sqrt{\log N}}{\sqrt{m_0N^2}}\bnorm{\M} +
		\frac{M\mu r\log N}{m_0}\binfnorm{\M}\bigg),	
	\end{align*}
	holds with probability at least $1-N^{-10}$.
\end{lem}
\begin{proof}
	Set $\La_0 = \{\l_1,\ldots,\l_{m_0}\}$.
	We use the same notation as in the proof of Lemma~\ref{lemtwo:general}. As a result, we have
	$$\binfnorm{\bmP_{\TT}(\frac{N^2}{m_0}\bmB_{\La_0} - \bmB)\M} = \max_{\aal\in\bOm}|\sum_{i=1}^{m_0} z_{\l_i,\aal}|,$$ where $z_{\l_i,\aal}$ is defined as in \eqref{eq:def:z:general}. And for any $\oom\in\bOm$,
	\begin{align}\label{est:zomegaalpha}
		|z_{\oom,\aal}| &= |\inp{\bmP_{\TT}(\frac{N^2}{m_0}\bmB_{\oom} - \frac{1}{m_0}\bmB)(\M)}{b_{\aal}\B_{\aal}}|\no\\
		&\leq \frac{2N^2}{m_0}\max_{\gga}|\inp{\bmP_{\TT}\bmB_{\oom}(\M)}{b_{\gga}\B_{\gga}}|\no\\
		&= \frac{2N^2}{m_0}\max_{\gga}|\inp{\bmP_{\TT}(b_{\oom}^{-1}\B_{\oom})}{b_{\gga}\B_{\gga}}| \cdot |\inp{\M}{b_{\oom}\B_{\oom}}|\no\\
		&\leq \frac{8N^2 M}{m_0}\max_{\gga}\fro{\bmP_{\TT}(\B_{\oom})}\fro{\bmP_{\TT}(\B_{\gga})} |\inp{\M}{b_{\oom}\B_{\oom}}|\no\\
		&\leq 8m_0^{-1}M\mu r |\inp{\M}{b_{\oom}\B_{\oom}}|,
	\end{align}
	where $(a)$ comes from the Lemma~\ref{lemma:incoh}. So from the definition of $\binfnorm{\cdot}$, we have
	\begin{align}
		|z_{\oom,\aal}| \leq 8m_0^{-1}M\mu r\binfnorm{\M}.
	\end{align}
	Now for any fixed $\aal\in\bOm$, we have $\EE\sum_{i=1}^{m_0} |z_{\l_i,\aal}|^2 = \frac{m_0}{N^2}\sum_{\oom}|z_{\oom,\aal}|^2$. So from \eqref{est:zomegaalpha}, we get,
	\begin{align}
		\sum_{\oom}|z_{\oom,\aal}|^2 \leq \sum_{\oom}(8m_0^{-1}M\mu r |\inp{\M}{b_{\oom}\B_{\oom}}|)^2\leq 64m_0^{-2}M^2\mu^2r^2\bnorm{\M}^2,
	\end{align}
	So we conclude that
	\begin{align}
		\EE\sum_{i=1}^{m_0} |z_{\l_i,\aal}|^2 \leq 64m_0^{-1}N^{-2}M^2\mu^2r^2\bnorm{\M}^2,
	\end{align}
	Now applying Bernstein concentration inequality and taking union bound yields that
	\begin{align*}
		\binfnorm{\bmP_{\TT}(\frac{N^2}{m_0}\bmB_{\Omega_0} - \bmB)(\M)}
		\leq &C\bigg(\sqrt{m_0^{-1}N^{-2}M^2\mu^2r^2\log N}\bnorm{\M} \\
		&\hspace{3cm} +m_0^{-1}M\mu r\log N\binfnorm{\M}\bigg)
	\end{align*}
	%	$$\binfnorm{\bmP_{\TT}(\frac{N^2}{m_0}\bmB_{\Omega_0} - \bmB)(\M)} \leq c_0\left(\frac{\mnorm\mu r\sqrt{\log(N^2)}}{\sqrt{m_0N^2}}\bnorm{\M} + \frac{\mnorm\mu r\log(N^2)}{m_0}\binfnorm{\M}\right)$$
	holds with probability exceeding $1-N^{-10}$ for some absolute constant $C>0$.
\end{proof}
Using the above lemmas, we get the following estimation.
\begin{lem}\label{lemma:contraction}
	Suppose $N^2q \geq CM\mu r\log(N_1+N_2)$ for some absolute constant $C > 0$, then we have with probability exceeding $1-2N^{-10}$ for all $1\leq i\leq L$,
	\begin{align*}
		&\quad\sqrt{\frac{\log(N_1+N_2)}{q}}\bnorm{\F_i} +\frac{\log(N_1+N_2)}{q}\binfnorm{\F_i}\\
		&\leq \frac{1}{2}\bigg(\sqrt{\frac{\log (N_1+N_2)}{q}}\bnorm{\F_{i-1}} + \frac{\log (N_1+N_2)}{q}\binfnorm{\F_{i-1}}\bigg).
	\end{align*}
\end{lem}
\begin{proof}
	From Lemma~\ref{lemtwo:general} and \ref{lemthree:general}, we have with failure probability at most $2N^{-10}$,
	\begin{align*}
		&\hspace{1cm}\sqrt{\frac{\log (N_1+N_2)}{q}}\bnorm{\F_i} +\frac{\log (N_1+N_2)}{q}\binfnorm{\F_i}\\
		&\leq C\sqrt\frac{\log (N_1+N_2)}{q}\left(\sqrt{\frac{M\mu r\log N}{N^2 q}}\bnorm{\F_{i-1}}+ \frac{1}{q}\sqrt{\frac{M\mu r}{N^2}}\log N\binfnorm{\F_{i-1}}\right)\\
		&\hspace{1cm}+C\frac{\log (N_1+N_2)}{q}\left(\frac{M\mu r\sqrt{\log N}}{N^2 \sqrt{q}}\bnorm{\F_{i-1}}+ \frac{M\mu r\log N}{N^2q}\binfnorm{\F_{i-1}}\right)\\
		&= C\left(\sqrt{\frac{M\mu r\log (N_1+N_2)}{N^2 q}} + \frac{M\mu r\log (N_1+N_2)}{N^2 q}\right)\bigg(\sqrt{\frac{\log (N_1+N_2)}{q}}\bnorm{\F_{i-1}}\\ &\hspace{10cm}+ \frac{\log (N_1+N_2)}{q}\binfnorm{\F_{i-1}}\bigg)\\
		&\leq \frac{1}{2}\bigg(\sqrt{\frac{\log (N_1+N_2)}{q}}\bnorm{\F_{i-1}} + \frac{\log (N_1+N_2)}{q}\binfnorm{\F_{i-1}}\bigg),
	\end{align*}
	where the last inequality holds since $N^2q\geq CM\mu r\log(N_1+N_2)$ and we use $\log N\leq \log(N_1+N_2)$.
\end{proof}
Taking a union bound and apply Lemma \ref{lemma:contraction}, and we get with probability exceeding $1- 2L^2N^{-10}$, for all $1\leq i\leq L$,
\begin{align*}
	&\quad\sqrt{\frac{\log (N_1+N_2)}{q}}\bnorm{\F_i} +\frac{\log (N_1+N_2)}{q}\binfnorm{\F_i}\\
	&\leq \frac{1}{2^i}\bigg(\sqrt{\frac{\log (N_1+N_2)}{q}}\bnorm{\F_{0}} + \frac{\log (N_1+N_2)}{q}\binfnorm{\F_{0}}\bigg).
\end{align*}
Now we go back to \eqref{eq:est:ptperpY:general}, and applying Lemma \ref{lemma:estofbm:general}, we obtain with probability exceeding $1-L(N_1+N_2)^{-10}$,
\begin{align*}
	\sum_{i=1}^L\op{(\frac{1}{q}\bmB_{\La_i} - \bmB)(\F_{i-1})} &\leq 	C\sum_{i=1}^L\left(\sqrt{\frac{\log(N_1+N_2)}{q}}\bnorm{\F_{i-1}} + \frac{\log(N_1+N_2)}{q}\binfnorm{\F_{i-1}}\right)\\
	&\leq 2C\left(\sqrt{\frac{\log(N_1+N_2)}{q}}\bnorm{\F_{0}} + \frac{\log(N_1+N_2)}{q}\binfnorm{\F_{0}}\right)
\end{align*}
Now we bound $\bnorm{\F_0}$ and $\binfnorm{\F_{0}}$. In fact,
for any $\oom\in\bOm$, from \cref{lemma:opnorm:B:general},
\begin{align*}
	|\inp{\F_0}{b_{\oom}\B_{\oom}}| \leq b_{\oom}\fro{\U}\fro{\B_{\oom}\V} \leq 8\sqrt{\frac{M\mu r^2}{KN_p n^2}},
\end{align*}
Therefore,
\begin{align*}
	\bnorm{\F_0}^2 &= \sum_{\oom\in\bOm} |\inp{\U\V^T}{b_{\oom}\B_{\oom}}|^2 \leq
	\max_{\oom\in\bOm}b_{\oom}\cdot\sum_{\oom\in\bOm} |\inp{\U\V^T}{\B_{\oom}}|^2
	\leq 16\frac{MN^2r}{KN_pn^2}.
\end{align*}
And
$
\binfnorm{\F_0} = \max_{\oom\in\bOm}b_{\oom}|\inp{\U\V^T}{\B_{\oom}}| \leq 8\sqrt{\frac{M\mu r^2}{KN_p n^2}}.
$
So we conclude
\begin{align*}
	\op{\bmP_{\TT}^{\perp}(\Y)}  \leq C\bigg(\sqrt{\frac{\log(N_1+N_2)}{q}}\sqrt{\frac{MN^2 r}{KN_pn^2}} + \frac{\log(N_1+N_2)}{q}\sqrt{\frac{M\mu r^2}{KN_pn^2}}\bigg) \leq \frac{1}{4},
\end{align*}
where the last inequality holds since $m\geq CLr\log(N_1+N_2)c_s^{-2}\frac{M}{K}\sqrt{\nu}$.

\hspace{1cm}

\noindent\textit{Estimation of $ \fro{\U\V^T - \bmP_{\TT}(\Y)}$.}
Notice that
\begin{align*}
	\U\V^T - \bmP_{\TT}(\Y) &= \bmP_{\TT}(\F_0) - \sum_{i=1}^{L} \bmP_{\TT}(\frac{1}{q}\bmB_{\La_i} + \bmB^{\perp})(\F_{i-1})\no\\
	&\overset{(a)}{=} \bmP_{\TT}(\F_1) - \sum_{i=2}^{L} \bmP_{\TT}(\frac{1}{q}\bmB_{\La_i} + \bmB^{\perp})(\F_{i-1})\no\\
	&=\ldots = \bmP_{\TT}(\F_{L}),	
\end{align*}
where $(a)$ is from the definition of $\F_i$. So it boils down to estimating $\fro{\bmP_{\TT}(\F_L)}$. In fact, we have
\begin{align}
	\fro{\bmP_{\TT}(\F_L)} &= \fro{\bmP_{\TT}(\bmB - \frac{1}{q}\bmB_{\La_L})\bmP_{\TT}(\F_{L-1})}\no\\
	&\leq \op{\bmP_{\TT}(\bmB - \frac{1}{q}\bmB_{\La_L})\bmP_{\TT}}\cdot\fro{\bmP_{\TT}(\F_{L-1})}\no\\
	&\leq \prod_{i=1}^L \op{\bmP_{\TT}(\bmB - \frac{1}{q}\bmB_{\La_i})\bmP_{\TT}} \cdot \fro{\bmP_{\TT}(\F_{0})}
\end{align}
Now from Lemma~\ref{lemma:estofptbomegapt-expectation}, as long as the sample size $m/L = |\La_i| \geq C\mu r\log(N_1+N_2)$ for some absolute constant $C>0$, we have $\op{\bmP_{\TT}(\bmB - \frac{1}{q}\bmB_{\La_i})\bmP_{\TT}} \leq \frac{1}{2e}$ with probability $1-L(N_1+N_2)^{-10}$ for all $1\leq i\leq L$. This implies
\begin{align}
	\fro{\bmP_{\TT}(\F_L)} \leq \left(\frac{1}{2e}\right)^L \fro{\bmP_{\TT}(\F_{0})} =  \left(\frac{1}{2e}\right)^L \sqrt{r} < \frac{1}{2N^4},
\end{align}
where in the last inequality we use $r\leq  N^2$ and $L = \lceil 4\log N\rceil$. And we finish the proof of the theorem when there is no noise.

\subsection{Noise Case}
Recall $\z$ is a solution to \cref{prob:noise} with $\left\|\bmP_{\La}\left(\z^*\right)-\y\right\|_2\leq\sqrt{m}\delta$.
We first estimate $\fro{\bmG(\z^*)-\bmG(\z)}$. From the definition of $\bmG$, we have
$$\sum_{k=1}^K\fro{[\bmH(\z)]_{\bOm_k} - [\bmH(\z^*)]_{\bOm_k}}^2 = \fro{\bmG(\z) - \bmG(\z^*)}^2.$$
Notice as constructed in section~\ref{sec:exact}, the dual certificate $\Y$ satisfying \eqref{dualcond:1:general}-\eqref{dualcond:3:general} exists with probability exceeding $1-3N^{-10}$. And from Lemma~\ref{lemma:estofptbomegapt-expectation}, we have \eqref{eq:lemma:uniqueness:cond1:general} holds with probability at least $1-N^{-10}$.
So with probability exceeding $1-4N^{-10}$, these assumptions hold and we prove the theorem assuming \eqref{eq:lemma:uniqueness:cond1:general}-\eqref{dualcond:3:general} hold.

We introduce some notations that will be used throughout. Denote by
$$\H^* = \bmG(\z^*), \H = \bmG(\z), \H_o = \bmG(\y), \X = \H-\H^*.$$
Then it boils down to estimating $\fro{\X}$.

Since $\X = \bmG(\z^*) - \bmG(\z)\in\range(\bmG)=\ker(\bmB^{\perp})$, we can write $\X = \bmB(\X) = \bmB_{\La}'(\X) + (\bmB - \bmB_{\La}')(\X)$, and we denote $(\bmB - \bmB_{\La}')(\X) =: \Z$.
Notice we have
\begin{align}\label{upperbound:X}
	\fro{\X} \leq \fro{\bmB_{\La}'(\X)} + \fro{\bmP_{\TT}(\Z)} + \fro{\bmP_{\TT}^{\perp}(\Z)}.
\end{align}
Using triangle inequality and we have,
\begin{align}
	\fro{\bmB_{\La}'(\X)} \leq \fro{\bmB_{\La}'(\H - \H_o)} + \fro{\bmB_{\La}'(\H^* - \H_o)}.
\end{align}
From the definition of $\delta$, $\fro{\bmP_{\La}(\z^* -\y)} = \sqrt{m}\delta$.
Since $\z$ is the minimizer to \cref{prob:noise}, we have $\fro{\bmP_{\La}(\z-\y)} \leq \sqrt{m}\delta$. Now from the proof of Lemma~\ref{lemma:opnorm:B:general}, each pixel appears at most $\frac{KN_pn^2}{N^2}$ times in the lifted matrix, we have the following two bounds,
\begin{align}
	\fro{\bmB_{\La}'(\H^*-\H_o)} \leq \sqrt{\frac{Kn^2N_pm}{N^2}}\delta~\text{and}~
	\fro{\bmB_{\La}'(\H-\H_o)} \leq \sqrt{\frac{Kn^2N_pm}{N^2}}\delta.
\end{align}
As a consequence, together with $\max\{\frac{n^2}{N^2},\frac{N_p}{N^2}\}\leq C_s$, we have
\begin{align}\label{eq:thm:noise:mainest:1}
	\fro{\bmB_{\La}'(\X)} \leq \sqrt{\frac{Kn^2N_pm}{N^2}}\delta.
	%	C_1\sqrt{KN^2m}\cdot\delta
\end{align}
%for some constant $C_1>0$ depending only on $C_s$.
Now we continue the proof by considering two cases.

\hspace{1cm}

\noindent\textit{Case 1: $\fro{\bmP_{\TT}(\Z)} \leq \frac{N^4}{2}\fro{\bmP_{\TT}^{\perp}(\Z)}$.}
Since $\Z\in \ker(\bmB_{\La})\cap\ker(\bmB^{\perp})$, we have similarly from \eqref{eq:est:H+Z:1:general} and $\W$ defined as in \eqref{eq:W:general} that
\begin{align}\label{eq:nuc:H+Z}
	\nuc{\H^*+\Z} &\geq \nuc{\H^*} +\frac{1}{2}\nuc{\bmP_{\TT}^{\perp}(\Z)} - \fro{\bmP_{\TT}(\Z)}\cdot\fro{\bmP_{\TT}(\W-\Y)} \no\\
	&\overset{(a)}{>}\nuc{\H^*} +\frac{1}{2}\nuc{\bmP_{\TT}^{\perp}(\Z)} -\frac{1}{2N^4}\fro{\bmP_{\TT}(\Z)}\no\\
	&\geq \nuc{\H^*} + \frac{1}{4}\fro{\bmP_{\TT}^{\perp}(\Z)},
\end{align}
where $(a)$ is from \eqref{dualcond:3:general}. Meanwhile, due to the optimality of $\H$, we have
\begin{align}\label{eq:nuc:H}
	\nuc{\H^*} \geq \nuc{\H} = \nuc{\H^*+\X}\geq \nuc{\H^* + \Z} - \nuc{\bmB_{\La}'(\X)}.
\end{align}
Now from \eqref{eq:nuc:H+Z} and \eqref{eq:nuc:H}, we get
\begin{align}\label{eq:thm:noise:mainest:2}
	\frac{1}{4}\fro{\bmP_{\TT}^{\perp}(\Z)} &\leq \nuc{\H^*+\Z} - \nuc{\H^*} \leq \nuc{\bmB_{\La}'(\X)}\leq \sqrt{K\max\{n^2,N_p\}}\fro{\bmB_{\La}'(\X)}\notag\\
	&\overset{(a)}{\leq}\sqrt{\max\{n^2,N_p\}n^2N_pmN^{-2}}K\delta,
\end{align}
where in $(a)$ we use \eqref{eq:thm:noise:mainest:1}.
Now we have
\begin{align}
	\fro{\left(\sqrt{\frac{N^2}{m}}\bmB_{\La} + \bmB^{\perp}\right)\bmP_{\TT}^{\perp}(\Z)}^2 &\overset{(a)}{=}\fro{\left(\sqrt{\frac{N^2}{m}}\bmB_{\La} + \bmB^{\perp}\right)\bmP_{\TT}(\Z)}^2\no\\
	&=\frac{N^2}{m}\inp{\bmB_{\La}\bmP_{\TT}(\Z)}{\bmB_{\La}\bmP_{\TT}(\Z)} + \inp{\bmB^{\perp}\bmP_{\TT}(\Z)}{\bmB^{\perp}\bmP_{\TT}(\Z)}\no\\
	&\geq\fro{\bmP_{\TT}(\Z)}^2 + \inp{\bmP_{\TT}(\Z)}{\bmP_{\TT}\left(\frac{N^2}{m}\bmB_{\La} - \bmB\right)\bmP_{\TT}(\Z)}\no\\
	&\overset{(b)}{\geq} \frac{1}{2}\fro{\bmP_{\TT}(\Z)}^2,
\end{align}
where $(a)$ is due to $\Z\in \ker(\bmB_{\La})\cap\ker(\bmB^{\perp})$ and $(b)$ is because of condition \eqref{eq:lemma:uniqueness:cond1:general}. From this estimation, we have
\begin{align}\label{eq:thm:noise:mainest:3}
	\fro{\bmP_{\TT}(\Z)}^2 \leq 2\fro{\left(\sqrt{\frac{N^2}{m}}\bmB_{\La} + \bmB^{\perp}\right)\bmP_{\TT}^{\perp}(\Z)}^2\overset{(a)}{\leq} \frac{8N^2}{m}\fro{\bmP_{\TT}^{\perp}(\Z)}^2\overset{(b)}{\leq} 128\max\{n^2,N_p\}n^2N_pK^2\delta^2,
\end{align}
where in $(a)$ we use that $\op{\sqrt{\frac{N^2}{m}}\bmB_{\La} + \bmB^{\perp}}\leq 2\sqrt{\frac{N^2}{m}}$ and in $(b)$ we use \eqref{eq:thm:noise:mainest:2}. Now combine \eqref{eq:thm:noise:mainest:1}, \eqref{eq:thm:noise:mainest:2} and \eqref{eq:thm:noise:mainest:3} and we arrive at the final estimation for \eqref{upperbound:X},
$$\fro{\X} \leq C\sqrt{\max\{n^2,N_p\}n^2N_p}K\delta,$$
for some absolute constant $C>0$.

\hspace{1cm}

\noindent\textit{Case 2: $\fro{\bmP_{\TT}(\Z)} > \frac{N^4}{2}\fro{\bmP_{\TT}^{\perp}(\Z)}$.}
As proved in Case 1 of proof of Lemma~\ref{lemma:uniqueness:general}, we can show $\Z = \bzero$ and thus $$\fro{\X} = \fro{\bmB_{\La}'(\X)} \leq \sqrt{\frac{Kn^2N_pm}{N^2}}\delta.$$
So combine these two cases and we get
$$\fro{\X} \leq C\sqrt{\max\{n^2,N_p\}n^2N_p}K\delta,$$
for some absolute constant $C>0$.

Now notice we have the following inequality:
$$\min_{\oom}c_{\oom}\cdot\fro{\z - \z^*}^2 \leq \fro{\bmG(\z) - \bmG(\z^*)}^2 = \fro{\X}^2.$$
This together with $\min_{\oom}c_{\oom} \geq \frac{KN_pn^2}{MN^2}$ and we get
$$\fro{\z^* - \z}^2\leq CMKN^2\max\{n^2,N_p\}\delta^2.$$
for some absolute constant $C > 0 $ and we finish the proof of the theorem.

\section{Numerical experiments}\label{sec:numerical}
In this section, we present numerical experiments related to the image inpainting problem based on the model \eqref{rankmin:noisefree}.
As the experiment result is not our main focus, we just show the result for completeness of this manuscript.
In fact, there has been a plethora of work implementing variations of our proposed algorithms (see e.g. \cite{S.Gu2014,H.Ji2010,L.Ma2017}).

For each of the images, 20\% of the pixels are randomly revealed.
In our algorithm, a reference image is required for grouping.
To this end, we use projected gradient descent to get a rough estimation from the following problem,
$$\min_{\z} \frac{1}{2}\|\bGamma \cdot \vec(\z)\|_{2}^2,~\text{s.t.}~\bmP_{\La}(\z) = \bmP_{\La}(\z^*),$$
where $\bGamma = \I_{n_1}\otimes \L_{n_2} +\L_{n_1}\otimes\I_{n_2}$ and $\L_{n}$ is the Laplacian matrix of size $n\times n$ with diagonal $1,2,\ldots,2,1$.
The reference might not be clear enough, but it is sufficient for grouping as our experiments will show.

We evenly choose $K$ reference patches and then find $N_p-1$ most similar patches in a neighborhood of the current patch according to the reference image and this is how we get $\bOm_k, k\in[K]$. Now we use alternating direction method of multipliers (ADMM) to solve \eqref{prob:noisefree}.

In the following we present the experiment result. We conduct image inpainting on the image: Pepper, Barbara, and Fingerprint. For each image, only 20 percent of the pixels are retained. The retained pixels are selected uniformly at random. The results are shown in Figure~\ref{fig:result}.

\begin{figure}
	\centering
	\subfigure[Barbara]{
		\includegraphics[width=1.3in]{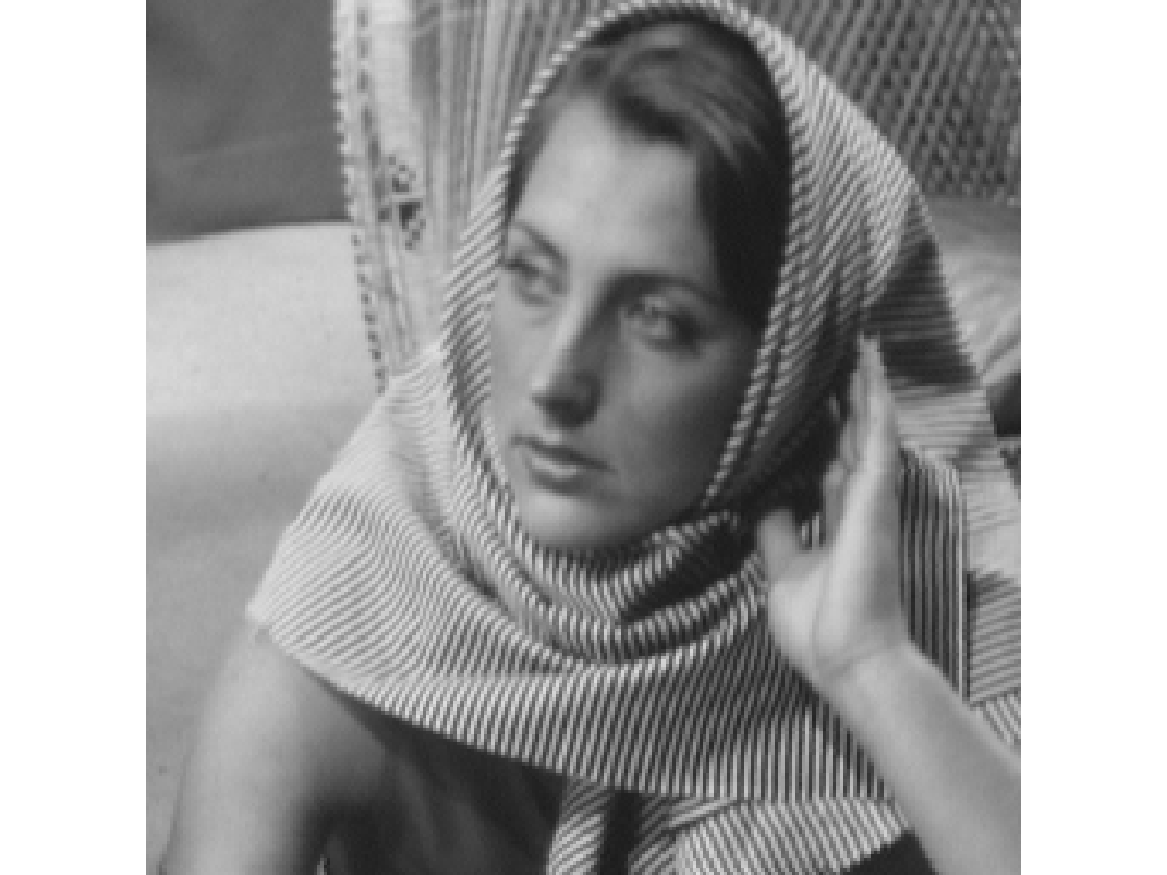}
	}
	\subfigure[20\% subsample]{
		\includegraphics[width=1.3in]{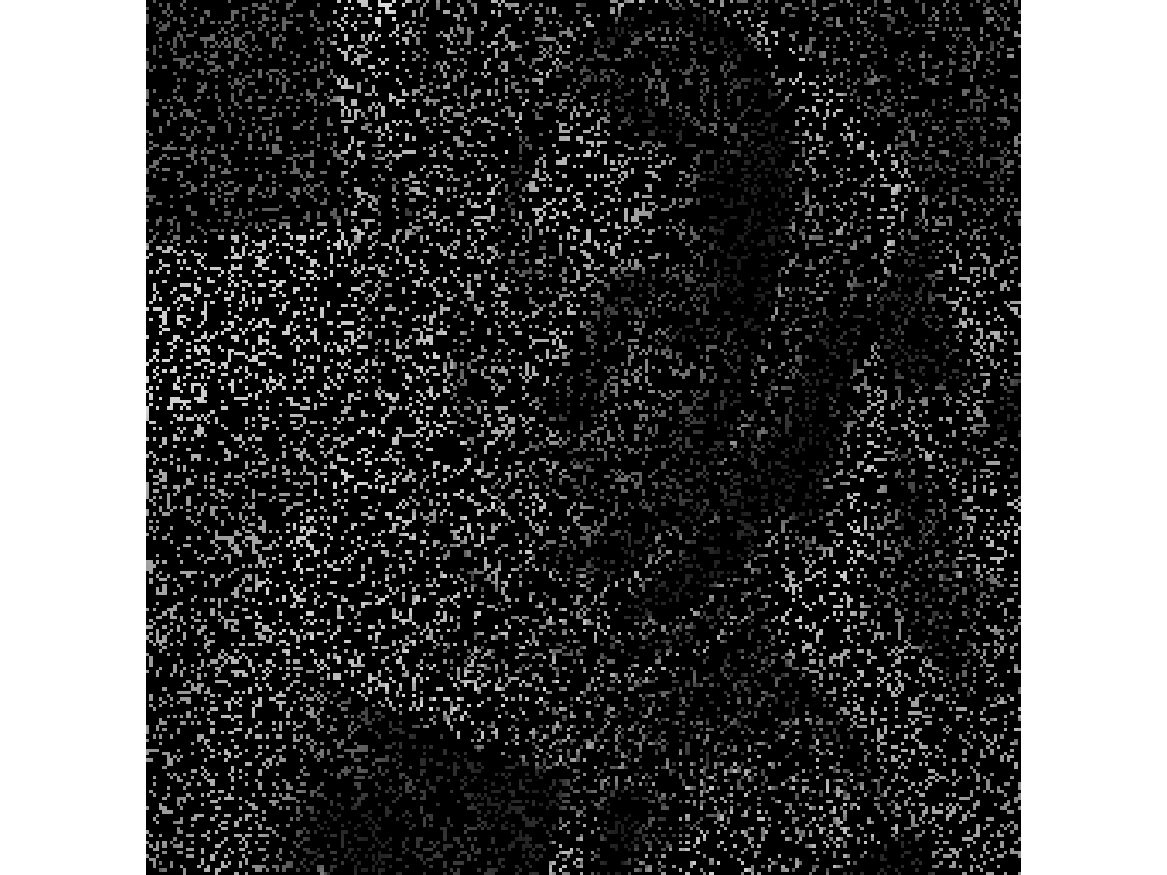}
	}
	\subfigure[25.06dB]{
		\includegraphics[width=1.3in]{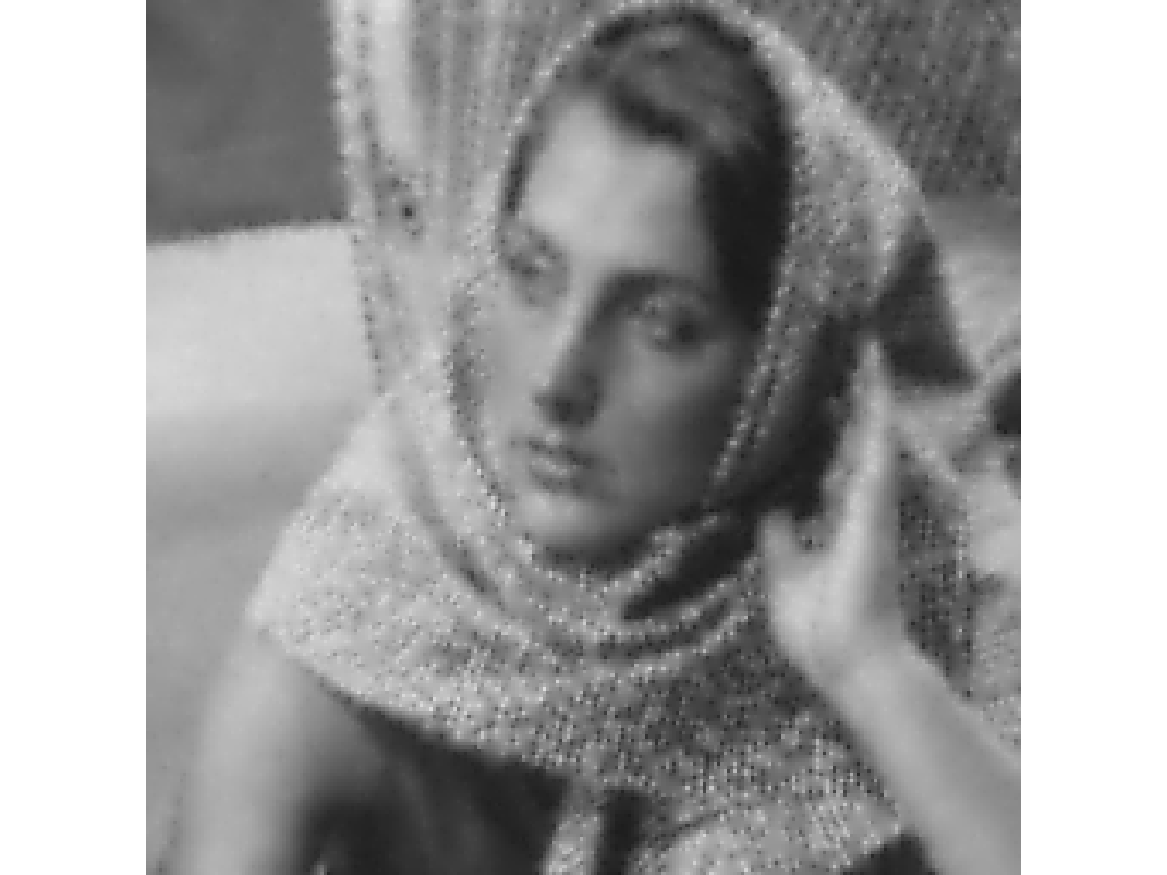}
	}
	\subfigure[28.48dB]{
		\includegraphics[width=1.3in]{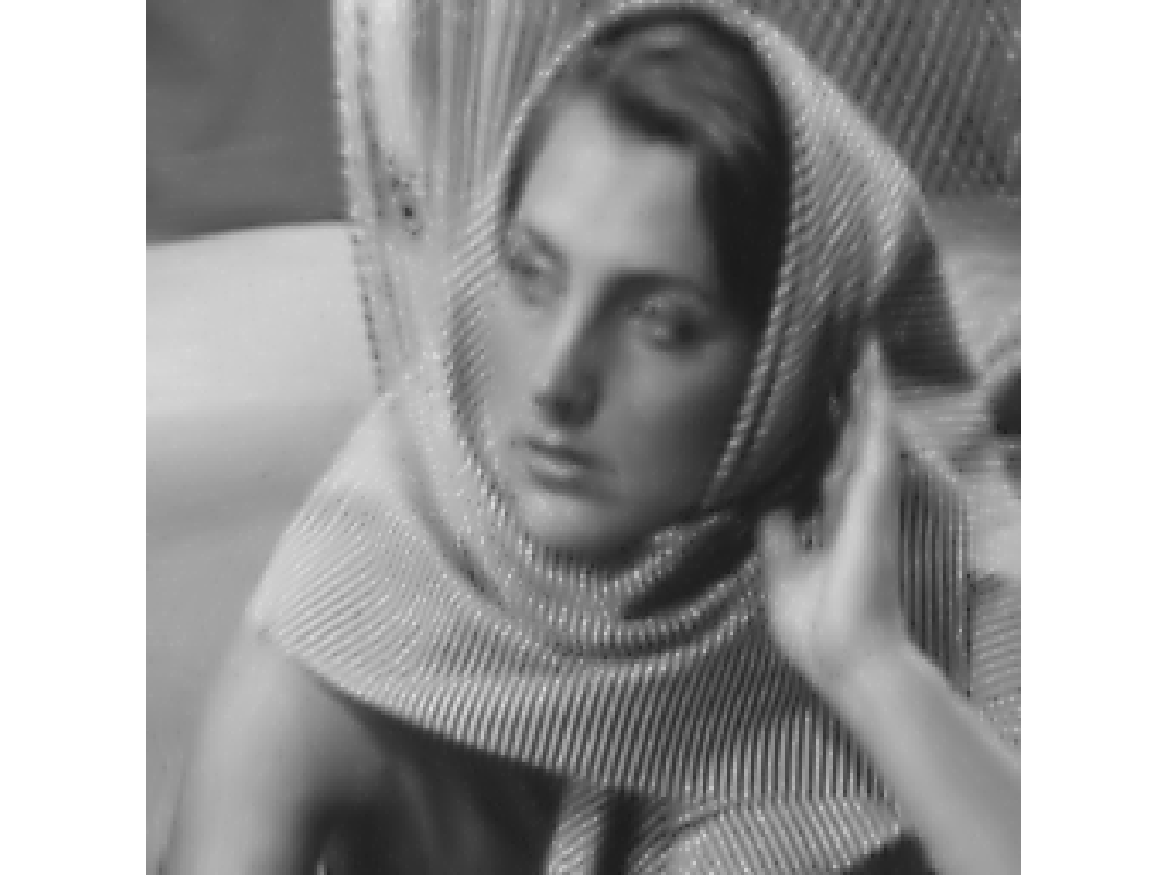}
	}
	\quad
	\subfigure[Fingerprint]{
		\includegraphics[width=1.3in]{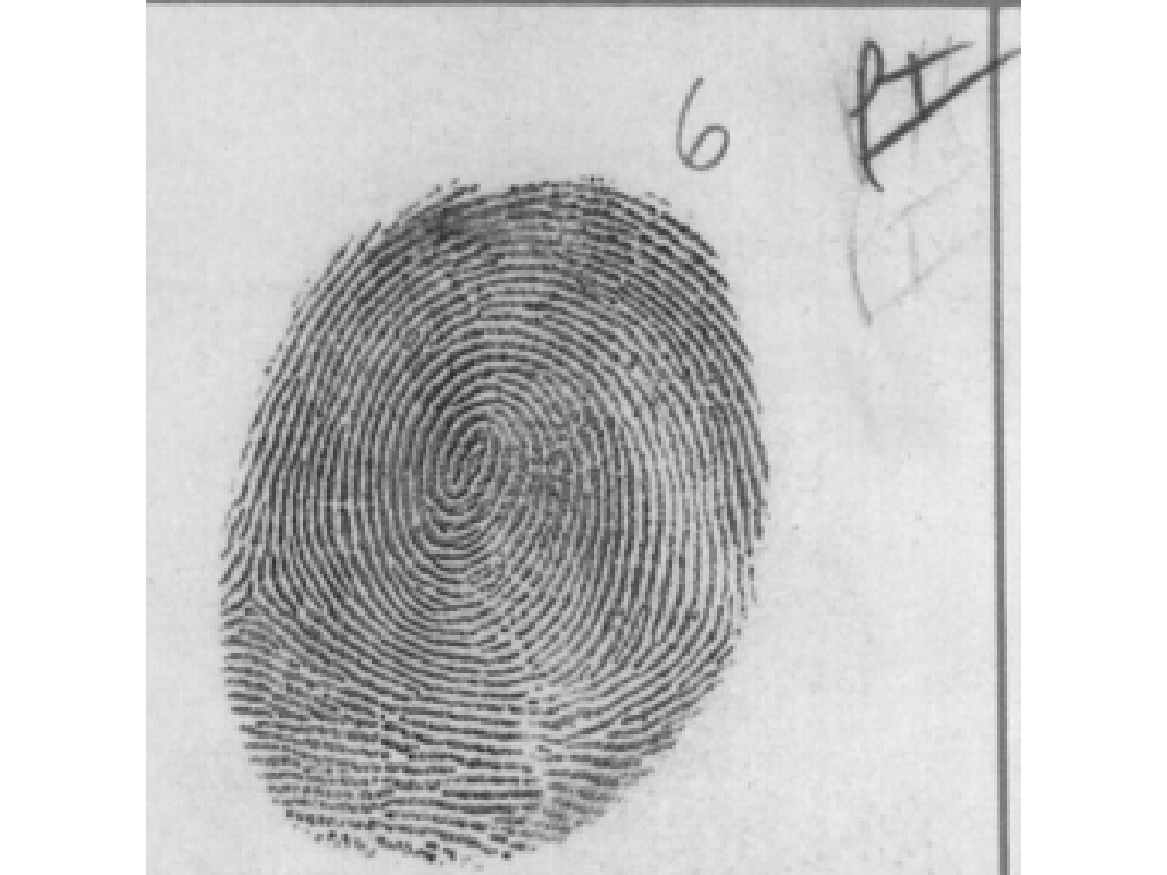}
	}
	\subfigure[20\% subsample]{
		\includegraphics[width=1.3in]{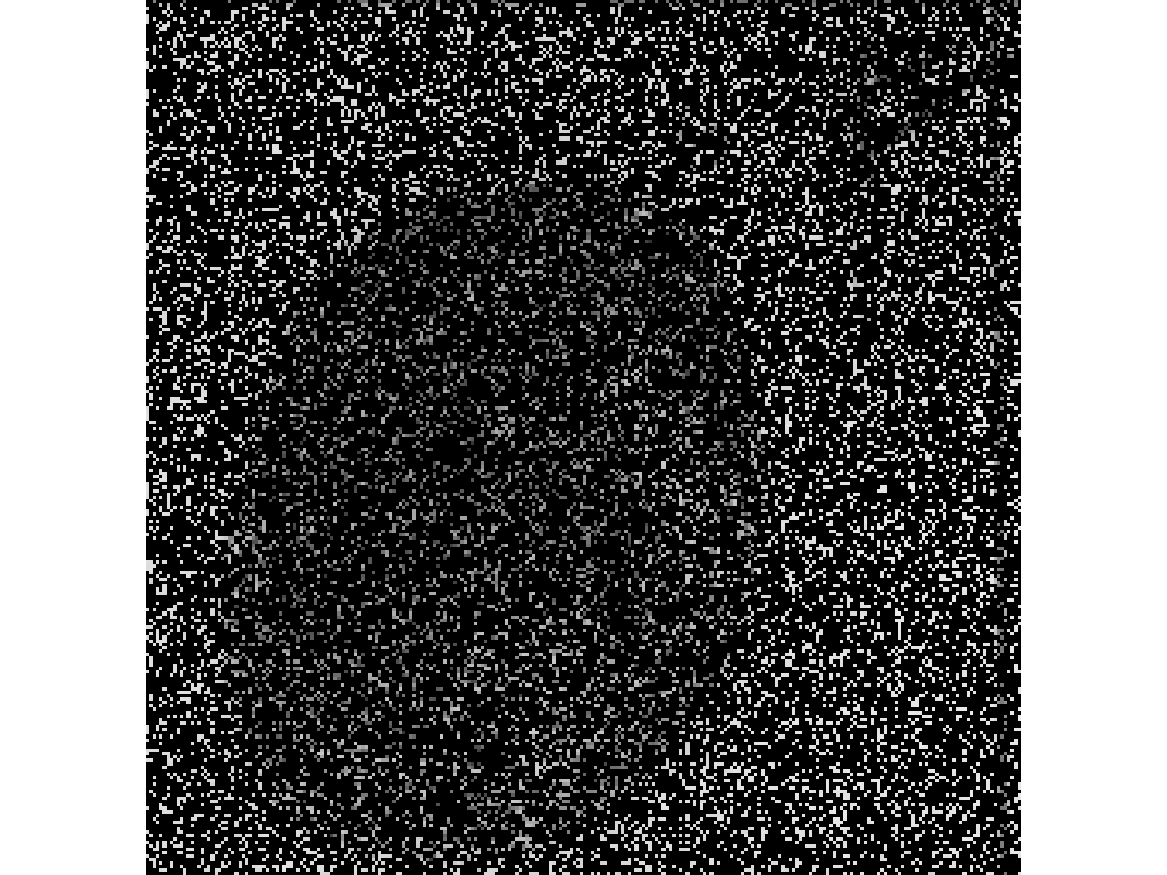}
	}
	\subfigure[22.34dB]{
		\includegraphics[width=1.3in]{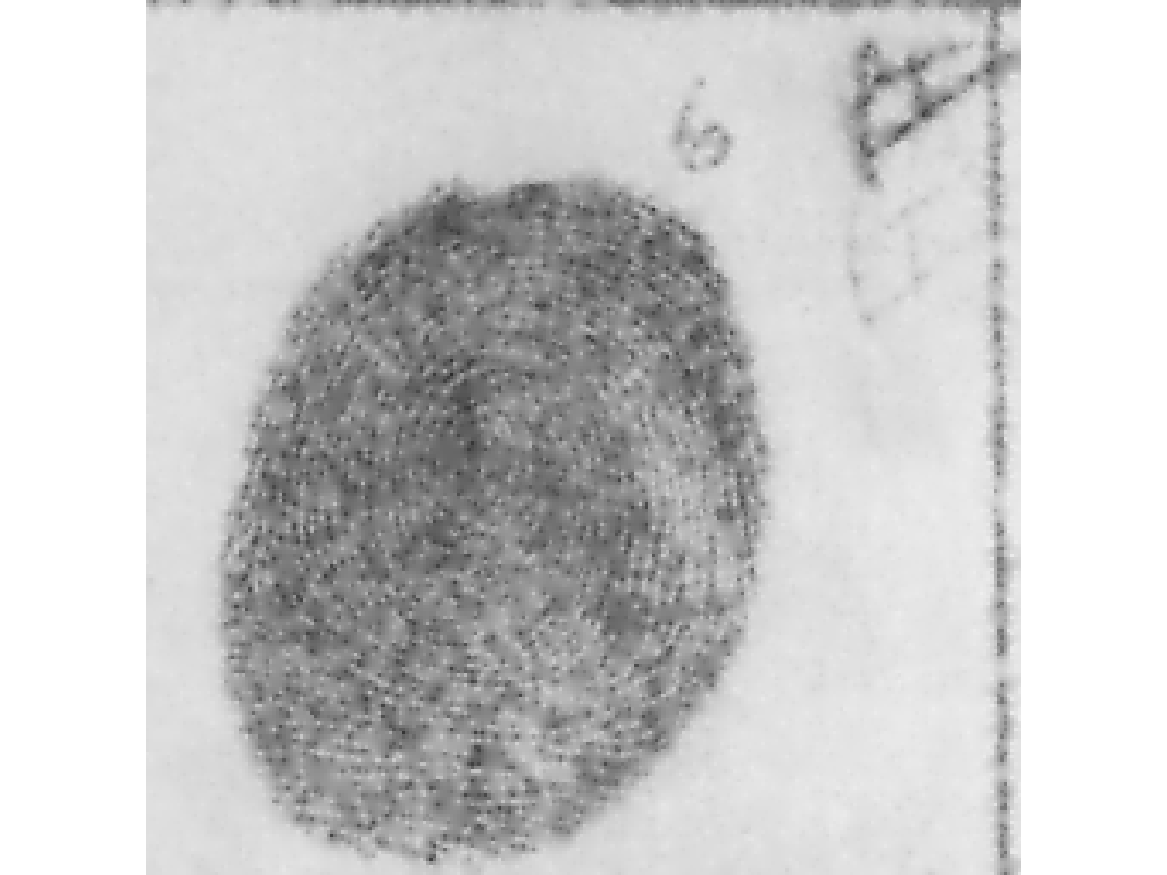}
	}
	\subfigure[24.04dB]{
		\includegraphics[width=1.3in]{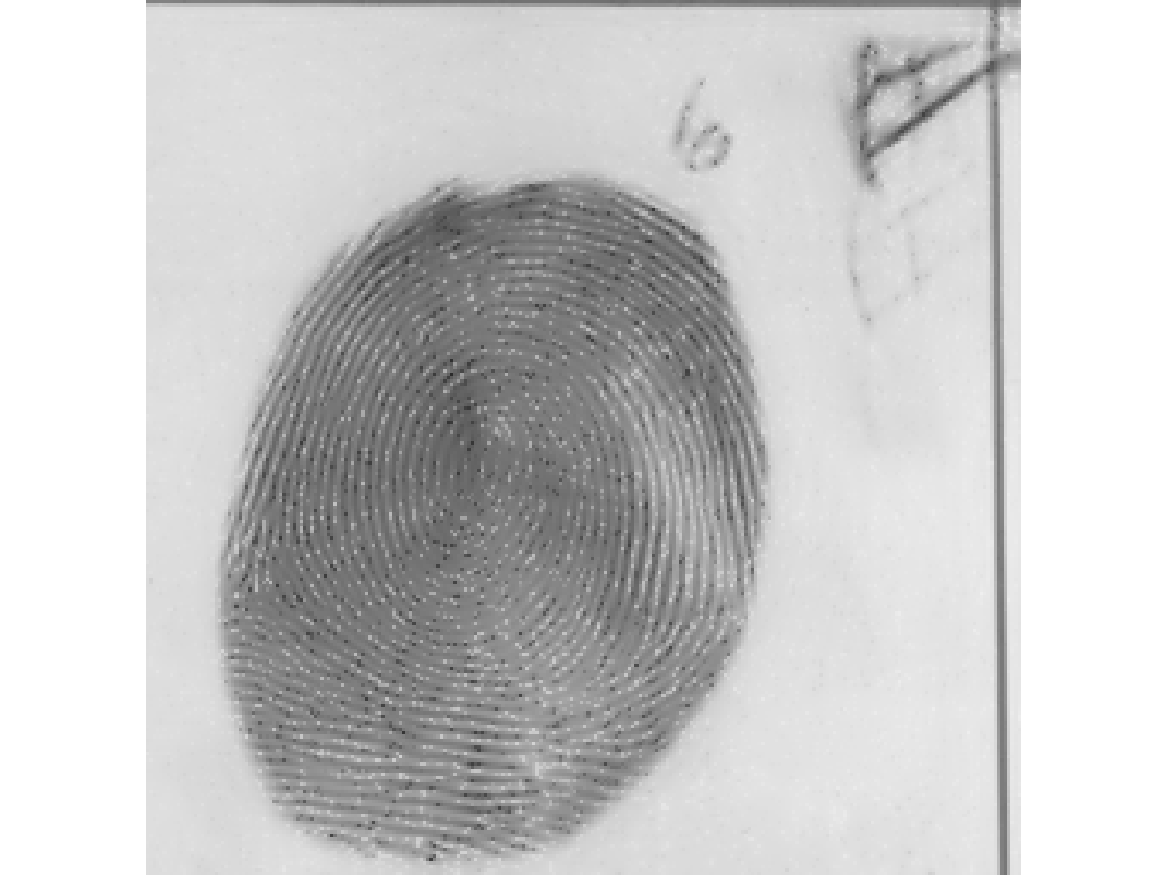}
	}
	\quad
	\subfigure[Lena]{
		\includegraphics[width=1.3in]{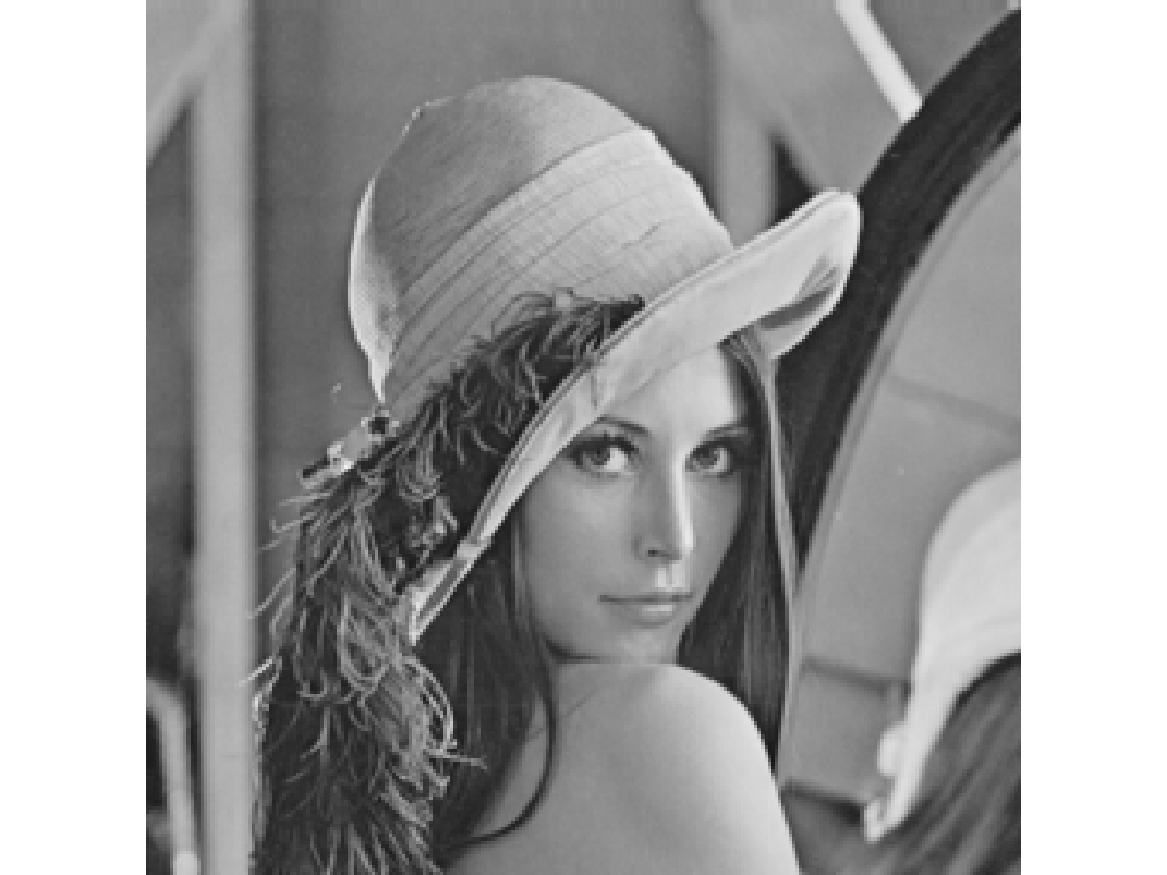}
	}
	\subfigure[20\% subsample]{
		\includegraphics[width=1.3in]{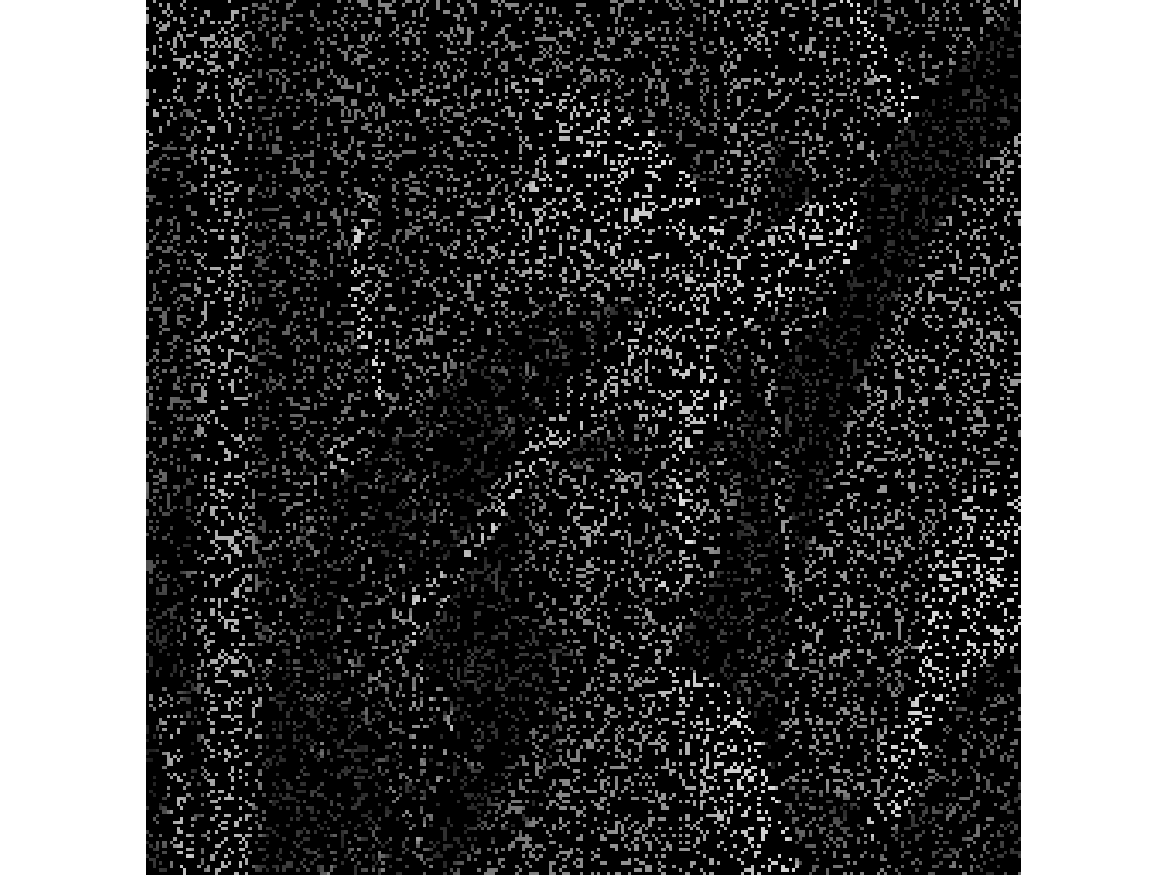}
	}
	\subfigure[26.63dB]{
		\includegraphics[width=1.3in]{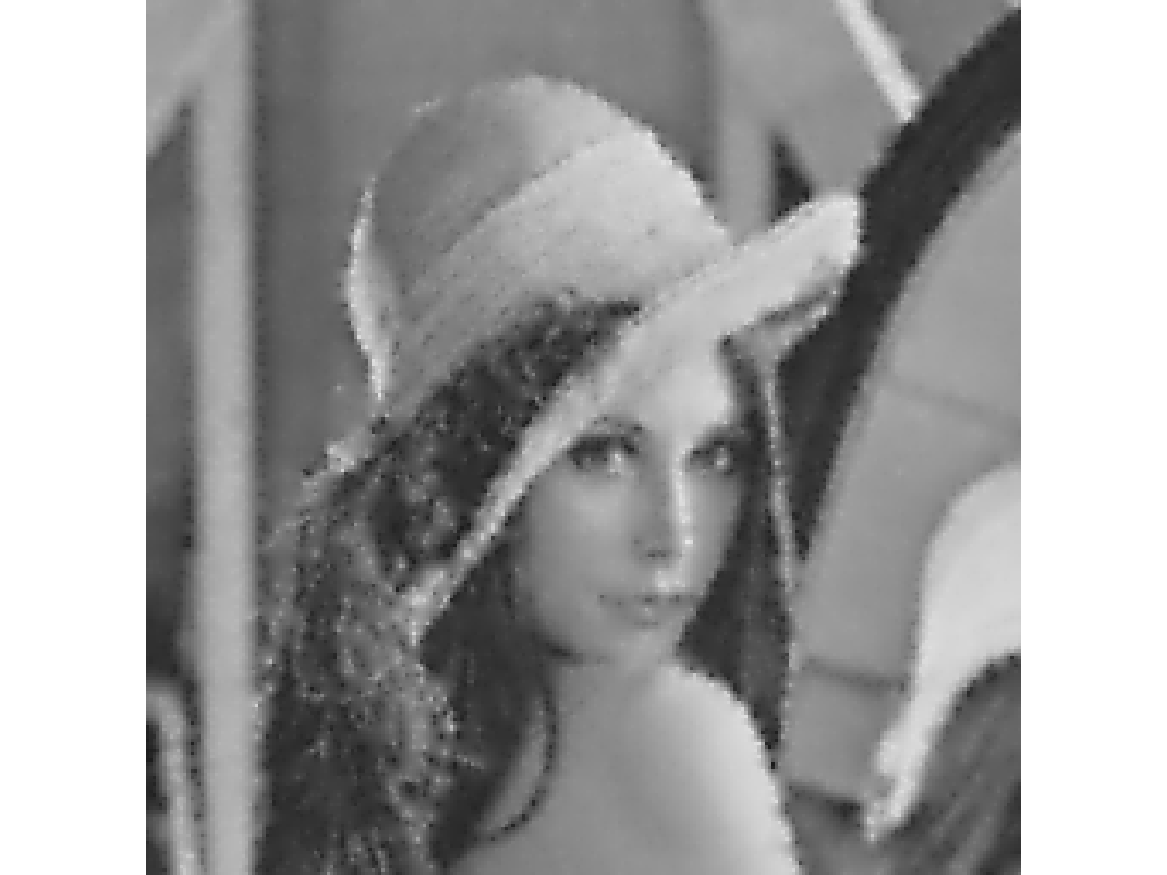}
	}
	\subfigure[28.37dB]{
		\includegraphics[width=1.3in]{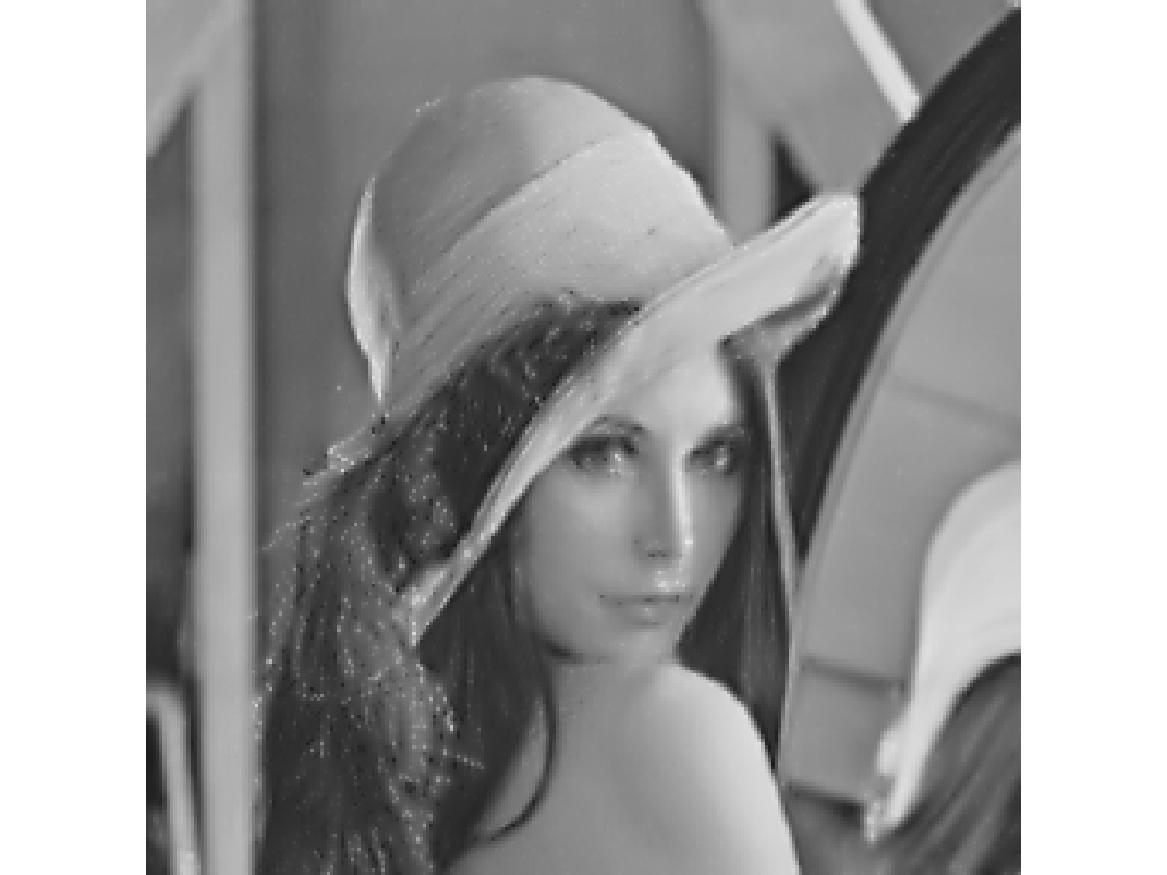}
	}
	\quad
	\subfigure[Pepper]{
		\includegraphics[width=1.3in]{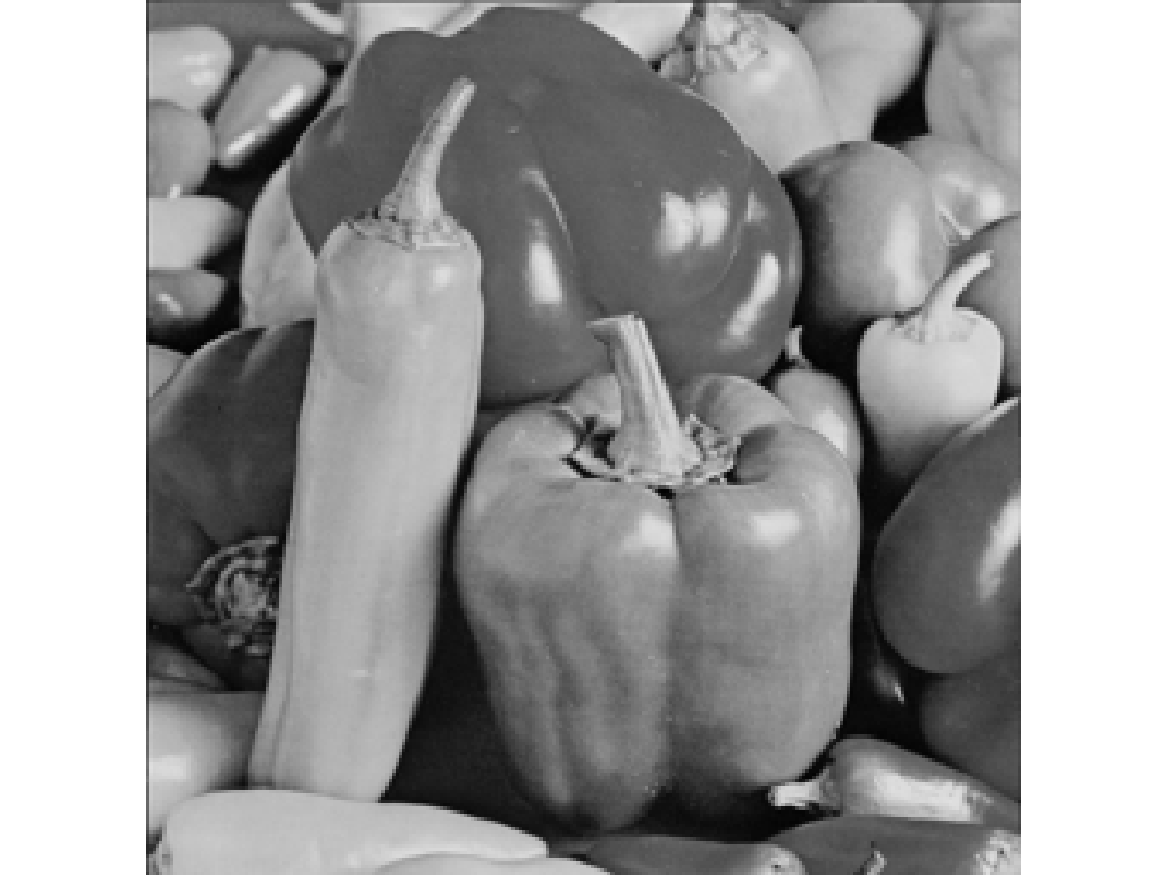}
	}
	\subfigure[20\% subsample]{
		\includegraphics[width=1.3in]{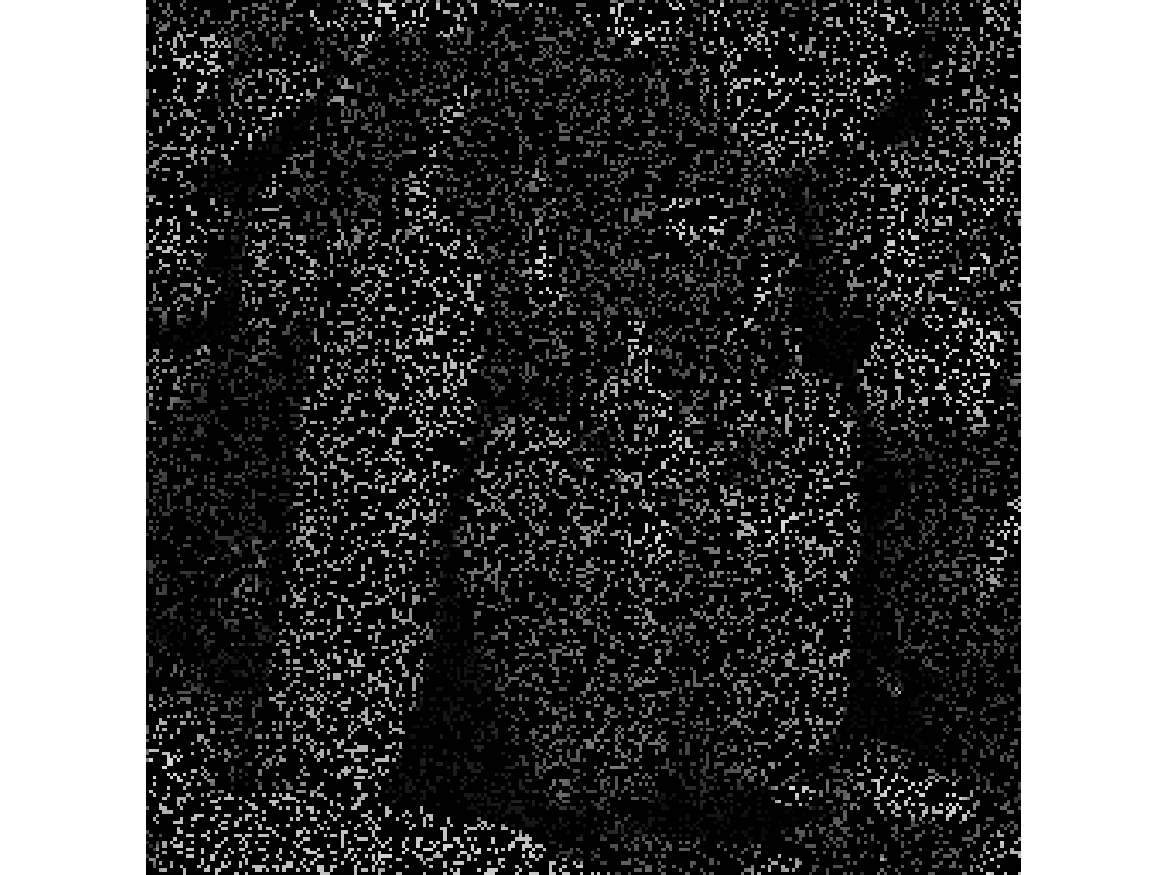}
	}
	\subfigure[26.48dB]{
		\includegraphics[width=1.3in]{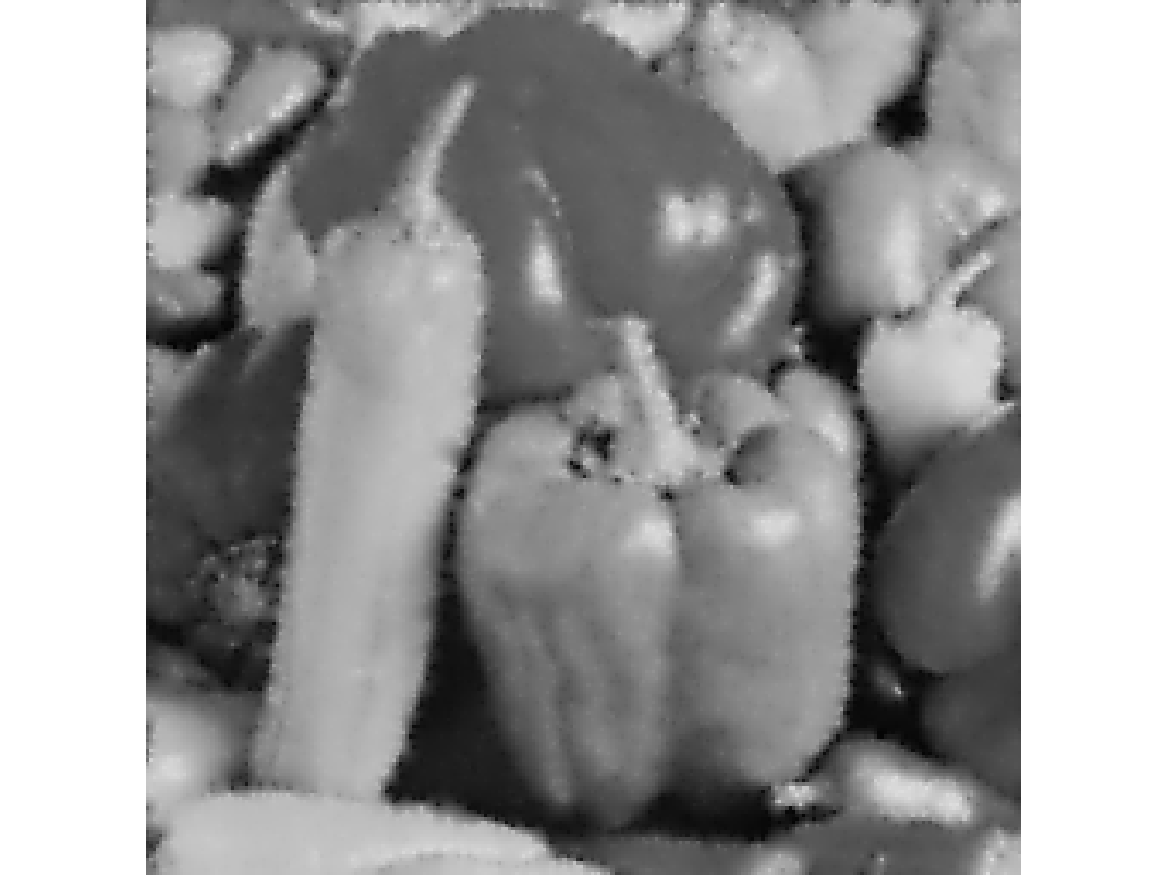}
	}
	\subfigure[28.39dB]{
		\includegraphics[width=1.3in]{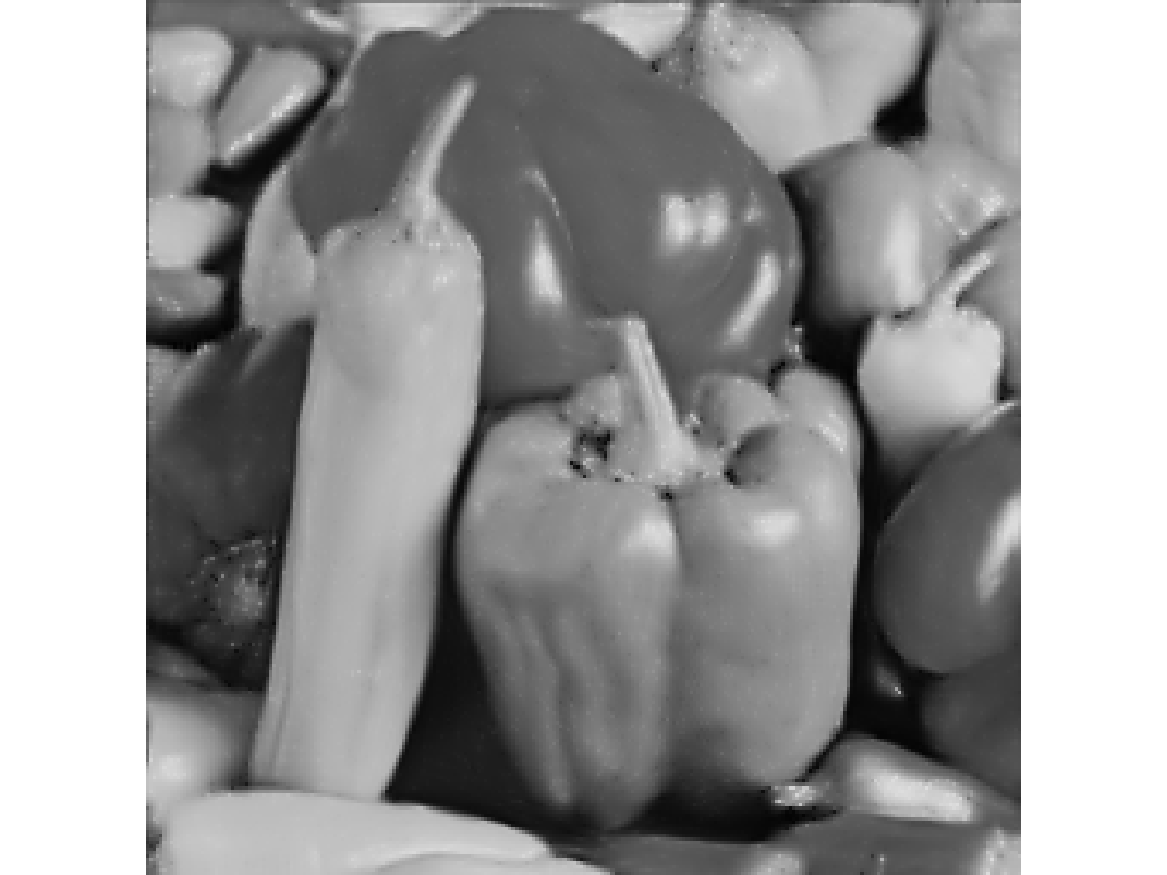}
	}
	\caption{The figures in the first column are the clear images; in the second column, the images with 20\% revealed are shown; in the third column, the reference images with its psnr are displayed; in the last column, we show the recovered images and its psnr.}
	\label{fig:result}
\end{figure}

\section{Conclusion and future works}\label{sec:conclusion}
In this paper, we provide theoretical guarantees for exact recovery of image models based on the patch based low rank prior, which has proven effective in various image restoration problems. Under mild incoherence conditions, we demonstrate that the image can be recovered accurately with only a few measurements. To the best of our knowledge, this is the first work to provide theoretical analysis of the patch-based method.

Moving forward, we would like to address several aspects in future research. Firstly, it remains unclear what is the approximation ability of the set $\MM_r \cap \text{range}(\bmG)$ approximate the true image.
Given a set of parameters much smaller than the image size ($N^2$), it is important to investigate the extent to which we can approximate a specific class of images.
This problem is closely related to variational autoencoders (VAE) \cite{kingma2019introduction}. However, VAE typically requires a large amount of training data for learning the generative model, whereas our focus is on a specific image.

Secondly, since our method restricts the lifted matrix to a specific subspace and ensures its low rank, it is worth considering subspace completion problems. Examples of such problems include matrix completion, where the subspace is the entire space, and the recovery of spectral sparse signals, which involves lifting to a Hankel matrix and considering the subspace of all Hankel matrices. A natural question arises: under what conditions on the sampling basis can we expect successful recovery of the object? While these examples have been extensively studied individually, their analysis heavily relies on the structure of the basis.

Thirdly, most existing patch-based methods treat patches as vectors, potentially leading to the loss of patch-specific information. It would be intriguing to explore low rank tensor completion problems instead. However, determining the appropriate rank model for our purpose remains a challenge. Due to the similarity between patches, we can expect the tensor to be well-approximated by a low rank tensor along the third direction. However, it is difficult to determine the rank along the other two directions. Therefore, it is essential to explore low-parameter models that best describe the underlying structure.

Lastly, the patch-based method can be viewed as a dictionary learning problem, aiming to find a dictionary $\D$ and a sparse coefficient matrix $\C$. In our case, the sparsity of $\C$ is enforced by lifting all the matrices $\bmH(\z^*)\big|_{\Omega_k}, k\in[K]$, to a block diagonal matrix. This implies that both $\C$ and $\D$ in this problem are block diagonal, with zero off-diagonal blocks, resulting in a sparse matrix $\C$. However, there are other possible constraints that can be imposed on the dictionary and coefficients to establish different models. Exploring these alternative decompositions and constraints is an interesting avenue for future research.

\bibliographystyle{plain}
\bibliography{References_PS}

\begin{thebibliography}{10}

\bibitem{M.Aharon2006}
M.~Aharon, M.~Elad, and A.~Bruckstein.
\newblock {$\rm K$-SVD}: an algorithm for designing overcomplete dictionaries
  for sparse representation.
\newblock {\em IEEE Trans. Signal Process.}, 54(11):4311--4322, Nov 2006.

\bibitem{altekruger2022patchnr}
Fabian Altekr{\"u}ger, Alexander Denker, Paul Hagemann, Johannes Hertrich,
  Peter Maass, and Gabriele Steidl.
\newblock Patchnr: Learning from small data by patch normalizing flow
  regularization.
\newblock {\em arXiv preprint arXiv:2205.12021}, 2022.

\bibitem{C.Bao2016}
C.~Bao, H.~Ji, Y.~Quan, and Z.~Shen.
\newblock Dictionary learning for sparse coding: algorithms and convergence
  analysis.
\newblock {\em IEEE Trans. Pattern Anal. Mach. Intell.}, 38(7):1356--1369, July
  2016.

\bibitem{K.Bredies2014}
K.~Bredies and M.~Holler.
\newblock Regularization of linear inverse problems with total generalized
  variation.
\newblock {\em J. Inverse Ill-Posed Probl.}, 22(6):871--913, 2014.

\bibitem{K.Bredies2010}
K.~Bredies, K.~Kunisch, and T.~Pock.
\newblock Total generalized variation.
\newblock {\em SIAM J. Imaging Sci.}, 3(3):492--526, 2010.

\bibitem{A.Buades2005}
A.~Buades, B.~Coll, and J.~M. Morel.
\newblock A review of image denoising algorithms, with a new one.
\newblock {\em Multiscale Model. Simul.}, 4(2):490--530, 2005.

\bibitem{J.F.Cai2010a}
J.~F. Cai, E.~J. Cand\`es, and Z.~Shen.
\newblock A singular value thresholding algorithm for matrix completion.
\newblock {\em SIAM J. Optim.}, 20(4):1956--1982, 2010.

\bibitem{J.F.Cai2010}
J.~F. Cai, R.~H. Chan, and Z.~Shen.
\newblock Simultaneous cartoon and texture inpainting.
\newblock {\em Inverse Probl. Imaging}, 4(3):379--395, 2010.

\bibitem{J.F.Cai2022}
J.~F. Cai, J.~K. Choi, J.~Li, and K.~Wei.
\newblock Image restoration: {S}tructured low rank matrix framework for
  piecewise smooth functions and beyond.
\newblock {\em Appl. Comput. Harmon. Anal.}, 56:26--60, 2022.

\bibitem{J.F.Cai2020}
J.~F. Cai, J.~K. Choi, and K.~Wei.
\newblock Data driven tight frame for compressed sensing {MRI} reconstruction
  via off-the-grid regularization.
\newblock {\em SIAM J. Imaging Sci.}, 13(3):1272--1301, 2020.

\bibitem{J.F.Cai2012}
J.~F. Cai, B.~Dong, S.~Osher, and Z.~Shen.
\newblock Image restoration: total variation, wavelet frames, and beyond.
\newblock {\em J. Amer. Math. Soc.}, 25(4):1033--1089, 2012.

\bibitem{J.F.Cai2014}
J.~F. Cai, H.~Ji, Z.~Shen, and G.~B. Ye.
\newblock Data-driven tight frame construction and image denoising.
\newblock {\em Appl. Comput. Harmon. Anal.}, 37(1):89--105, 2014.

\bibitem{J.F.Cai2009/10}
J.~F. Cai, S.~Osher, and Z.~Shen.
\newblock Split {B}regman methods and frame based image restoration.
\newblock {\em Multiscale Model. Simul.}, 8(2):337--369, 2009/10.

\bibitem{J.F.Cai2018a}
J.~F. Cai, T.~Wang, and K.~Wei.
\newblock Spectral compressed sensing via projected gradient descent.
\newblock {\em SIAM J. Optim.}, 28(3):2625--2653, 2018.

\bibitem{J.F.Cai2019}
J.~F. Cai, T.~Wang, and K.~Wei.
\newblock Fast and provable algorithms for spectrally sparse signal
  reconstruction via low-rank {H}ankel matrix completion.
\newblock {\em Appl. Comput. Harmon. Anal.}, 46(1):94--121, 2019.

\bibitem{cai2016robust}
Jian-Feng Cai, Xiaobo Qu, Weiyu Xu, and Gui-Bo Ye.
\newblock Robust recovery of complex exponential signals from random gaussian
  projections via low rank hankel matrix reconstruction.
\newblock {\em Applied and computational harmonic analysis}, 41(2):470--490,
  2016.

\bibitem{E.Candes2006}
E.~Cand{\`e}s, L.~Demanet, D.~Donoho, and L.~Ying.
\newblock Fast discrete curvelet transforms.
\newblock {\em Multiscale Model. Simul.}, 5(3):861--899, 2006.

\bibitem{E.J.Candes2009}
E.~J. Cand\`es and B.~Recht.
\newblock Exact matrix completion via convex optimization.
\newblock {\em Found. Comput. Math.}, 9(6):717--772, 2009.

\bibitem{A.Chambolle1997}
A.~Chambolle and P.~L. Lions.
\newblock Image recovery via total variation minimization and related problems.
\newblock {\em Numer. Math.}, 76(2):167--188, 1997.

\bibitem{R.H.Chan2005}
R.~H. Chan, C.~W. Ho, and M.~Nikolova.
\newblock Salt-and-pepper noise removal by median-type noise detectors and
  detail-preserving regularization.
\newblock {\em IEEE Trans. Image Process.}, 14(10):1479--1485, Oct 2005.

\bibitem{T.F.Chan2006}
T.~F. Chan, J.~Shen, and H.~M. Zhou.
\newblock Total variation wavelet inpainting.
\newblock {\em J. Math. Imaging Vision}, 25(1):107--125, 2006.

\bibitem{Y.Chi2019}
Y.~Chi, Y.~M. Lu, and Y.~Chen.
\newblock Nonconvex optimization meets low-rank matrix factorization: an
  overview.
\newblock {\em IEEE Trans. Signal Process.}, 67(20):5239--5269, Oct 2019.

\bibitem{J.K.Choi2020}
J.~K. Choi, B.~Dong, and X.~Zhang.
\newblock An edge driven wavelet frame model for image restoration.
\newblock {\em Appl. Comput. Harmon. Anal.}, 48(3):993--1029, 2020.

\bibitem{K.Dabov2006}
K.~Dabov, A.~Foi, V.~Katkovnik, and K.~Egiazarian.
\newblock Image denoising with block-matching and 3{D} filtering.
\newblock In Nasser~M. Nasrabadi, Syed~A. Rizvi, Edward~R. Dougherty, Jaakko~T.
  Astola, and Karen~O. Egiazarian, editors, {\em Image Processing: Algorithms
  and Systems, Neural Networks, and Machine Learning}, volume 6064, pages 354
  -- 365. International Society for Optics and Photonics, SPIE, 2006.

\bibitem{B.Dong2017}
B.~Dong, Q.~Jiang, and Z.~Shen.
\newblock Image restoration: wavelet frame shrinkage, nonlinear evolution
  {PDE}s, and beyond.
\newblock {\em Multiscale Model. Simul.}, 15(1):606--660, 2017.

\bibitem{B.Dong2017a}
B.~Dong, Z.~Shen, and P.~Xie.
\newblock Image restoration: a general wavelet frame based model and its
  asymptotic analysis.
\newblock {\em SIAM J. Math. Anal.}, 49(1):421--445, 2017.

\bibitem{M.Elad2006}
M.~Elad and M.~Aharon.
\newblock Image denoising via sparse and redundant representations over learned
  dictionaries.
\newblock {\em IEEE Trans. Image Process.}, 15(12):3736--3745, Dec 2006.

\bibitem{Gross2011}
D.~Gross.
\newblock Recovering low-rank matrices from few coefficients in any basis.
\newblock {\em IEEE Trans. Inform. Theory}, 57(3):1548--1566, March 2011.

\bibitem{gross2011recovering}
David Gross.
\newblock Recovering low-rank matrices from few coefficients in any basis.
\newblock {\em IEEE Transactions on Information Theory}, 57(3):1548--1566,
  2011.

\bibitem{S.Gu2014}
S.~Gu, L.~Zhang, W.~Zuo, and X.~Feng.
\newblock Weighted nuclear norm minimization with application to image
  denoising.
\newblock In {\em 2014 IEEE Conference on Computer Vision and Pattern
  Recognition}, pages 2862--2869, June 2014.

\bibitem{B.Han2014}
B.~Han and Z.~Zhao.
\newblock Tensor product complex tight framelets with increasing
  directionality.
\newblock {\em SIAM J. Imaging Sci.}, 7(2):997--1034, 2014.

\bibitem{Hou_2016_CVPR}
Le~Hou, Dimitris Samaras, Tahsin~M. Kurc, Yi~Gao, James~E. Davis, and Joel~H.
  Saltz.
\newblock Patch-based convolutional neural network for whole slide tissue image
  classification.
\newblock In {\em Proceedings of the IEEE Conference on Computer Vision and
  Pattern Recognition (CVPR)}, June 2016.

\bibitem{H.Ji2010}
H.~Ji, C.~Liu, Z.~Shen, and Y.~Xu.
\newblock Robust video denoising using low rank matrix completion.
\newblock In {\em 2010 IEEE Computer Society Conference on Computer Vision and
  Pattern Recognition}, pages 1791--1798, June 2010.

\bibitem{H.Ji2011}
H.~Ji, Z.~Shen, and Y.~Xu.
\newblock Wavelet based restoration of images with missing or damaged pixels.
\newblock {\em East Asian J. Appl. Math.}, 1(2):108--131, May 2011.

\bibitem{H.Ji2017}
H.~Ji, Z.~Shen, and Y.~Zhao.
\newblock Directional frames for image recovery: multi-scale discrete {G}abor
  frames.
\newblock {\em J. Fourier Anal. Appl.}, 23(4):729--757, 2017.

\bibitem{kingma2019introduction}
Diederik~P Kingma, Max Welling, et~al.
\newblock An introduction to variational autoencoders.
\newblock {\em Foundations and Trends{\textregistered} in Machine Learning},
  12(4):307--392, 2019.

\bibitem{G.Kutyniok2011}
G.~Kutyniok and W.~Q. Lim.
\newblock Compactly supported shearlets are optimally sparse.
\newblock {\em J. Approx. Theory}, 163(11):1564--1589, 2011.

\bibitem{laus2017nonlocal}
Friederike Laus, Mila Nikolova, Johannes Persch, and Gabriele Steidl.
\newblock A nonlocal denoising algorithm for manifold-valued images using
  second order statistics.
\newblock {\em SIAM Journal on Imaging Sciences}, 10(1):416--448, 2017.

\bibitem{lebrun2013nonlocal}
Marc Lebrun, Antoni Buades, and Jean-Michel Morel.
\newblock A nonlocal bayesian image denoising algorithm.
\newblock {\em SIAM Journal on Imaging Sciences}, 6(3):1665--1688, 2013.

\bibitem{L.Ma2017}
L.~Ma, L.~Xu, and T.~Zeng.
\newblock Low rank prior and total variation regularization for image
  deblurring.
\newblock {\em J. Sci. Comput.}, 70(3):1336--1357, 2017.

\bibitem{Mallat2008}
S.~Mallat.
\newblock {\em A Wavelet Tour of Signal Processing, Third Edition: The Sparse
  Way}.
\newblock Academic Press, 3rd edition, 2008.

\bibitem{S.Osher2017}
S.~Osher, Z.~Shi, and W.~Zhu.
\newblock Low dimensional manifold model for image processing.
\newblock {\em SIAM J. Imaging Sci.}, 10(4):1669--1690, 2017.

\bibitem{S.Osher2003}
S.~Osher, A.~Sol\'{e}, and L.~Vese.
\newblock Image decomposition and restoration using total variation
  minimization and the {$H^{-1}$} norm.
\newblock {\em Multiscale Model. Simul.}, 1(3):349--370, 2003.

\bibitem{Recht2011}
B.~Recht.
\newblock A simpler approach to matrix completion.
\newblock {\em J. Mach. Learn. Res.}, 12:3413--3430, 2011.

\bibitem{B.Recht2010}
B.~Recht, M.~Fazel, and P.~A. Parrilo.
\newblock Guaranteed minimum-rank solutions of linear matrix equations via
  nuclear norm minimization.
\newblock {\em SIAM Rev.}, 52(3):471--501, 2010.

\bibitem{ronneberger2015u}
Olaf Ronneberger, Philipp Fischer, and Thomas Brox.
\newblock U-net: Convolutional networks for biomedical image segmentation.
\newblock In {\em Medical Image Computing and Computer-Assisted
  Intervention--MICCAI 2015: 18th International Conference, Munich, Germany,
  October 5-9, 2015, Proceedings, Part III 18}, pages 234--241. Springer, 2015.

\bibitem{L.I.Rudin1992}
L.~I. Rudin, S.~Osher, and E.~Fatemi.
\newblock Nonlinear total variation based noise removal algorithms.
\newblock {\em Phys. D}, 60(1-4):259--268, 1992.
\newblock Experimental mathematics: computational issues in nonlinear science
  (Los Alamos, NM, 1991).

\bibitem{H.Schaeffer2013}
H.~Schaeffer and S.~Osher.
\newblock A low patch-rank interpretation of texture.
\newblock {\em SIAM J. Imaging Sci.}, 6(1):226--262, 2013.

\bibitem{sharma2017patch}
Atharva Sharma, Xiuwen Liu, Xiaojun Yang, and Di~Shi.
\newblock A patch-based convolutional neural network for remote sensing image
  classification.
\newblock {\em Neural Networks}, 95:19--28, 2017.

\bibitem{ulyanov2018deep}
Dmitry Ulyanov, Andrea Vedaldi, and Victor Lempitsky.
\newblock Deep image prior.
\newblock In {\em Proceedings of the IEEE conference on computer vision and
  pattern recognition}, pages 9446--9454, 2018.

\bibitem{L.A.Vese2003}
L.~A. Vese and S.~J. Osher.
\newblock Modeling textures with total variation minimization and oscillating
  patterns in image processing.
\newblock volume~19, pages 553--572. 2003.
\newblock Special issue in honor of the sixtieth birthday of Stanley Osher.

\bibitem{wei2016guarantees}
Ke~Wei, Jian-Feng Cai, Tony~F Chan, and Shingyu Leung.
\newblock Guarantees of riemannian optimization for low rank matrix completion.
\newblock {\em arXiv preprint arXiv:1603.06610}, 2016.

\bibitem{zhang2019patch}
Yushu Zhang, Hongbo Lin, Yue Li, and Haitao Ma.
\newblock A patch based denoising method using deep convolutional neural
  network for seismic image.
\newblock {\em IEEE Access}, 7:156883--156894, 2019.

\bibitem{zoran2011learning}
Daniel Zoran and Yair Weiss.
\newblock From learning models of natural image patches to whole image
  restoration.
\newblock In {\em 2011 international conference on computer vision}, pages
  479--486. IEEE, 2011.

\end{thebibliography}
\end{document}